\documentclass[twocolumn]{autart}    % Enable this line and disable the 
\usepackage{graphicx}          % Include this line if your 
\usepackage{graphics} % for pdf, bitmapped graphics files
\usepackage{epsfig} % for postscript graphics files
\usepackage{times} % assumes new font selection scheme installed
\usepackage{amsmath} % assumes amsmath package installed
\usepackage{amssymb}  % assumes amsmath package installed
\usepackage{mathtools}
\usepackage[short, nocomma]{optidef}
\usepackage{tensor}
\usepackage{tikz}
\usetikzlibrary{decorations.pathreplacing,calc}
\usepackage{pgfplots}
\pgfplotsset{compat=newest}
\usepackage{url}

\usepackage{tikz}
\usetikzlibrary{calc,patterns,angles,quotes}
\usetikzlibrary{spy}
\usepackage{colortbl}%

\usepackage{eqparbox}
\usepackage{svg}
\usepackage{makecell}
\usepackage{boldline}
\usepackage{cite}
\usepackage{xcolor}
\usepackage{algorithm}
\usepackage{algpseudocode}
\usepackage{graphicx}
\usepackage{bm}
\usepackage{siunitx}
\usepackage{pgfplots}
\usepackage{pgfplotstable}
\usetikzlibrary{backgrounds}

\newcommand{\rev}[1]{#1}

\newcommand{\iatgroupmail}{%
  {%\ttfamily
  \begin{tabular}[t]{@{}l@{}}
    \{hendrik.alsmeier;~maik.pfefferkorn;\\
    ~rolf.findeisen\}@iat.tu-darmstadt.de
  \end{tabular}}%
}

\mathtoolsset{showonlyrefs} % Comment out to get manual numbering.

\begin{document}

\begin{frontmatter}

\title{Certified Neural Approximation of Robust Nonlinear Model Predictive Control}

\thanks{Email: \iatgroupmail}

\author[Darmstadt]{Hendrik Alsmeier},\quad
\author[Darmstadt]{Maik Pfefferkorn},\quad
\author[Darmstadt]{Rolf Findeisen*}

\address[Darmstadt]{Control and Cyber-Physical Systems Laboratory, Technical University of Darmstadt, Landgraf-Georg-Str. 4,
64283 Darmstadt, Germany}

%%%%%%%%%%%%%%%%%%%%%%%%%%%%%%%%%%%%%%%%%%%%%%%%%%%%%%%%%%%%%%%%%%%%%%%%%%%%%%%%
\begin{abstract}
Model predictive control is a high-performance control method, but it incurs comparatively high online computational costs across many control tasks and is challenging to realize in low-level hardware settings. To address this, we present a framework for generating neural network controllers that provably inherit the safety and stability guarantees of a robust model predictive policy. In contrast to approaches based on post-verification or additional online safety filters, our method provides constructive, a priori design conditions that ensure preservation of recursive feasibility and stability before deployment. We treat the approximation error as a bounded input disturbance, which is addressed in the MPC design. This allows the neural network controller to directly inherit robustness from the underlying model predictive controller, which enables the transfer of guarantees to the approximation. The design process is an offline, \rev{constructive} synthesis pipeline and can be implemented once Lipschitz bounds are specified. We further demonstrate in simulation that the learned controller reproduces the closed-loop behavior of the original controller while reducing online computation.
\end{abstract}

\end{frontmatter}
%%%%%%%%%%%%%%%%%%%%%%%%%%%%%%%%%%%%%%%%%%%%%%%%%%%%%%%%%%%%%%%%%%%%%%%%%%%%%%%%
\section{Introduction}

Model predictive control (MPC) has been well-established as a reliable feedback control approach that is based on solving an underlying optimal control problem repeatedly and online in a receding horizon fashion \cite{rawlings2017,gruene2016,lucia2016}.
As such, it can handle a variety of systems, can include input, state, and output constraints, and ensure near-optimal closed-loop behavior with respect to a given control objective. These properties have led to its widespread adoption in industry and research.
However, the need to solve the optimal control problem online can impose a computational burden, posing implementation challenges in resource-constrained environments, such as embedded hardware, due to the demanding nature of the optimization process.

To address this challenge, various approaches have been proposed. For example, in the linear–quadratic case with linear dynamics, quadratic costs, and polytopic state and input constraints, explicit MPC yields a feedback law that is piecewise affine in the state. This structural property enables explicit MPC to compute the feedback law offline in closed form, thereby avoiding the need to solve an optimization problem online \cite{alessio2009,zeilinger2011}. However, the number of regions of the piecewise affine law grows, in the worst case, exponentially with the total number of inequality constraints, which in turn increases with the horizon length and the number of constrained states and inputs. Consequently, both the offline computation and memory complexity become rapidly intractable as the problem size grows.

\rev{This motivates the development of approximation- and learning-based approaches that mimic the implicit MPC feedback law~\cite{allgoewer2018a,Trimpe2020,chen2018deepmpc,maddalena2020,karg2020efficient,paulson2020,drummond2022,lucia2018deep,lueken2023,rose2023,tokmak2025,Hose2025}. Among the available function approximators, neural networks (NNs) are particularly attractive and employed in this work due to their universal approximation capabilities~\cite{hornik1991,hornik1993} and their ability to maintain a fixed evaluation effort with growing state dimension, input dimension, and dataset sizes. These properties have motivated numerous neural imitation-learning approaches for approximating MPC feedback laws. Nevertheless, the choice of function approximator alone does not determine how closed-loop guarantees are obtained. Existing approaches can be distinguished more coherently by how the approximation error is incorporated into the control design and how compliance with a prescribed error bound is certified.}

\rev{One prominent guarantee-transfer paradigm consists of two steps. First, the MPC is designed to be robust against a prescribed approximation error, which is interpreted as a bounded disturbance acting on the system state. Second, an explicit approximation of the resulting MPC feedback law is constructed such that its error satisfies the prescribed bound. An early NN-based formulation of this approach was presented in~\cite{allgoewer2018a}, where Hoeffding's inequality is used to validate statistically whether the learned controller satisfies the admissible input-error bound. If validation fails, additional samples are generated, and the learning procedure is repeated. Building on this framework,~\cite{Trimpe2020} extends the approach to nonlinear setpoint tracking and introduces a less conservative validation criterion that jointly accounts for approximation error and model mismatch. In~\cite{alsmeier24}, statistical post-validation is replaced by a deterministic condition that combines Lipschitz properties of the MPC feedback law and the NN with the coverage of the training data.}

\rev{This disturbance-based guarantee-transfer principle is independent of the chosen function approximator. In~\cite{rose2023}, a Gaussian process is used to learn the MPC feedback law, and compliance with the required error bound is established through tailored training-data selection. Recently,~\cite{tokmak2025} employs automatic kernel interpolation to construct an explicit surrogate whose approximation error is uniformly bounded by design, under the assumption that the feedback law belongs to the associated reproducing kernel Hilbert space (RKHS) with a known norm bound. 
% Similarly,~\cite{rose2025} develops a kernel-based interpolation scheme that selectively samples the MPC law on the closed-loop reachable set and uses Monte Carlo reachability analysis to certify safety and closed-loop performance. 
Similarly,~\cite{rose2025} uses kernel interpolation on the closed-loop reachable set and Monte Carlo reachability analysis to certify safety and closed-loop performance.
% These approaches share the disturbance-based control argument of the NN methods above, differing only in controller representation and error certification.
These approaches share the disturbance-based guarantee-transfer principle of the NN methods above, but differ in the choice of approximator, the associated regularity assumptions, and the mechanism used to certify the approximation error.}

\rev{Other NN-based approaches obtain guarantees through direct error certification, structural restrictions, or online correction mechanisms. In~\cite{drummond2022}, semidefinite programming is used to bound the worst-case difference between the NN and MPC feedback laws, while~\cite{paulson2020} uses an optimization-based projection to preserve constraint satisfaction and stability. 
% Karg et al.~\cite{karg2020efficient} build on constrained NN architectures and explicit MPC approximations~\cite{chen2018deepmpc,maddalena2020} by representing explicit MPC laws through deep piecewise-affine networks with theoretical approximation guarantees. 
Building on constrained NN approximations of explicit MPC laws~\cite{chen2018deepmpc,maddalena2020}, deep piecewise-affine networks provide approximation guarantees by design~\cite{karg2020efficient}. For nonlinear MPC,~\cite{Hose2025} validates NN proposals online and uses an MPC-intrinsic fallback when needed, while~\cite{lucia2018deep,lueken2023} learn MPC laws from trajectory data, with~\cite{lueken2023} improving data efficiency through gradient-informed Sobolev training.}

\rev{Against this background, we develop a constructive framework for the certified neural approximation of robust nonlinear MPC. The framework adopts the disturbance-based viewpoint of~\cite{allgoewer2018a,Trimpe2020,rose2023,tokmak2025}, but derives a deterministic certificate from dataset coverage and Lipschitz bounds, thereby preserving closed-loop guarantees without relying on statistical validation or requiring membership of the MPC feedback law in a prescribed RKHS with a corresponding norm bound. 
% Directly building on~\cite{alsmeier24}, we relax the requirement of global Lipschitz continuity of the MPC feedback law and establish the transfer of guarantees under local regularity conditions along the closed-loop trajectory. 
We generalize the framework of~\cite{alsmeier24} to MPC feedback laws that need not be globally Lipschitz continuous, establishing guarantee transfer under local regularity conditions along the closed-loop trajectory.
To this end, we first design a robust MPC for a prescribed disturbance bound and derive conditions ensuring that a neural network trained on MPC state-input data satisfies this bound, thereby preserving safety and stability without post-verification or auxiliary safety mechanisms. We further validate these guarantees in closed-loop operation. Our main contributions are:}
\begin{enumerate}
    \item[(i)] A priori conditions are derived that link the Lipschitz properties of the MPC feedback law and the coverage of the training data to certified bounds on the approximation error.
    \item[(ii)] \rev{A constructive offline procedure is developed to obtain a neural network approximator that satisfies these bounds, enabling systematic training of the network.}
    \item[(iii)] It is shown that, under these conditions, the learned neural network controller inherits the recursive feasibility and stability properties of the robust MPC, ensuring safe and reliable closed-loop operation.
\end{enumerate}

The remainder of this article is structured as follows. In Section \ref{sec:problem}, we formulate the problem and introduce the robust MPC scheme. 
Section \ref{sec:lmpcaeb} introduces the basics of the neural networks used, details our theoretical results on the neural network approximation of the underlying MPC, and further includes an overview of the constructive procedure enabling the generation of the neural networks in line with the previously established theory.
We illustrate the effectiveness of our proposed approach in Section \ref{sec:ex} in simulation using a 4-dimensional continuously stirred tank reactor example before concluding in Section~\ref{sec:conc}.

%%%%%%%%%%%%%%%%%%%%%%%%%%%%%%%%%%%%%%%%%%%%%%%%%%%%%%%%%%%%%%%%%%%%%%%%%%%%%%%%
\section{Robust MPC as a Basis for Certified Controller Approximation}\label{sec:problem}

We first show that the neural network approximation error enters the system as a bounded input disturbance. Using Lipschitz continuity of the system dynamics, this input disturbance can be equivalently reformulated as a bounded additive state disturbance. This reformulation allows us to cast the closed-loop system in a standard robust control setting. Based on this representation, we introduce the robust nonlinear MPC scheme that forms the basis for transferring safety and stability guarantees to the neural network controller. The presentation focuses on a robust MPC formulation for additive state disturbances, which is sufficient for our subsequent analysis. We conclude the section by recalling the key closed-loop properties of the robust MPC, in particular, recursive feasibility and input-to-state stability.

\subsection{System Description and Control Task}\label{sec:proform}

We consider nonlinear discrete-time systems given by 
\begin{equation}
    x_{k+1} = f(x_k,u_k), \label{eq:systemdesc}
\end{equation}
where $x_k \in \mathbb{R}^{n_x}$ and $u_k \in \mathbb{R}^{n_u}$ describe the system state and input, respectively, and $k \in \mathbb{N}_0$ denotes the discrete time index. The transition map is given by $f: \mathbb{R}^{n_x} \times \mathbb{R}^{n_u} \rightarrow \mathbb{R}^{n_x}$ and captures the dynamics of the system.
The objective is to stabilize the system \eqref{eq:systemdesc} \rev{at the origin (w.l.o.g.)} while satisfying input and state constraints given by
\begin{equation}\label{eq:constraints}
    \forall k \in \mathbb{N}:~ u_k \in \mathcal{U} \subset \mathbb{R}^{n_u},~ x_k \in \mathcal{X} \subset \mathbb{R}^{n_x}.
\end{equation}

While model predictive control is well-suited for this task, it requires the online solution of a possibly large-scale optimization problem, typically through iterative methods. Thus, we aim to control the system \eqref{eq:systemdesc} by a low-complexity neural network approximation $\kappa_\mathrm{NN}$ of the model predictive feedback law $u = \kappa(x)$ to reduce computational demands or enable formal verification of the explicit feedback law.
However, the use of the neural network introduces an approximation error
\begin{equation}
    e\rev{(x)} = \kappa(x) - \kappa_\mathrm{NN}(x),\label{eq:approx_error}
\end{equation}
which manifests as an additive disturbance on the control input.
If this disturbance is appropriately accounted for in the MPC design, the safety and stability guarantees can be transferred to the neural network controller.

To do so, we rely on a pre-specified bound on the neural network's approximation error, which we ensure through suitable data selection and training.
That is, we assume that we know $\epsilon > 0$ \rev{such that the approximation error in~\eqref{eq:approx_error} satisfies} $\forall x \in \mathcal{X}:~ \| e \| \leq \epsilon$.\footnote{\rev{Unless otherwise specified, any vector $p$-norm with $p\geq 1$ may be used. In finite dimensions these norms are equivalent, e.g., $\|x\|_\infty \leq \|x\|_p \leq n_x^{1/p}\|x\|_\infty$, although the choice of norm may affect the conservatism of the resulting bound.}}
We then reformulate the resulting bounded input disturbance as a bounded additive disturbance acting on the state of the system \eqref{eq:systemdesc} under the feedback law $\kappa$.
This state disturbance is then addressed within a robust MPC formulation, allowing guarantees established for the closed-loop system under the robust MPC to carry over to any approximate controller that satisfies the prescribed error bound.
The crucial step is to design a neural network that provably adheres to the predefined approximation error.
To this end, we exploit Lipschitz continuity properties of both the MPC feedback law and the neural network, together with a data coverage criterion for the training data with respect to the domain $\mathcal{X}$.

We make the following assumptions on the properties of the system and constraint sets.

\begin{assum}\label{ass:system_dynamics}
    The transition map $f: \mathbb{R}^{n_x} \times \mathbb{R}^{n_u} \to \mathbb{R}^{n_x}$ is twice continuously differentiable in $x$ and $u$ on $\mathbb{R}^{n_x} \times \mathbb{R}^{n_u}$, denoted by $f \in C^2(\mathbb{R}^{n_x} \times \mathbb{R}^{n_u}, \mathbb{R}^{n_x})$.
    Moreover, it holds that $f(0, 0) = 0$.
\end{assum}

Note that the differentiability property in Assumption \ref{ass:system_dynamics} is needed to apply a known result from parametric optimization regarding the structural properties of the MPC feedback law, which constitutes the first of two key components enabling the transfer of guarantees from the robust MPC to the neural network controller. In practice, many commonly considered systems satisfy this assumption.

\begin{assum}\label{ass:constraint_sets}
    The sets $\mathcal{U}$ and $\mathcal{X}$ are compact and given by  $\mathcal{U} := \{ g_i^u(u) \leq 0 \mid i = 1, \ldots, n_c^u \}$ and $\mathcal{X} := \{ g_i^x(x) \leq 0 \mid i = 1, \ldots, n_c^x \}$ with twice continuously differentiable functions $g_i^u(u) \in C^2(\mathbb{R}^{n_u}, \mathbb{R})$ and $g_i^x(x) \in C^2(\mathbb{R}^{n_x}, \mathbb{R})$. Moreover, we have that $(0, 0) \in \mathrm{int} (\mathcal{X} \times \mathcal{U})$.
\end{assum}

As before, Assumption~\ref{ass:constraint_sets} is necessary to ensure the desired structural properties of the optimal control problem underlying the MPC feedback law.
However, most practical MPC formulations employ polytopic or ellipsoidal sets for $\mathcal{X}$ and $\mathcal{U}$, for which Assumption \ref{ass:constraint_sets} holds. % by default.

\subsection{From Input Approximation Error to State Disturbance}

Most robust MPC formulations account for state disturbances.
Therefore, we reformulate the approximation error as an equivalent state disturbance, enabling the application of a broad class of robust MPC \rev{formulations that provide a priori guarantees of closed-loop robustness to the neural network controller's approximation error}.

Assumptions \ref{ass:system_dynamics} and \ref{ass:constraint_sets} imply Lipschitz continuity of the transition map $f$ with respect to both $x$ and $u$ on the compact domain $\mathcal{X} \times \mathcal{U}$.
That is, there exist finite constants $L_{f,x} > 0$ and $L_{f,u} > 0$ such that
\begin{subequations}
\begin{align}
\label{eq:lipschitz_f_u}
\| f(x, u) - f(x, u') \| & \leq L_{f,u} \| u - u' \|, \\
    \label{eq:lipschitz_f_x} \| f(x, u) - f(x', u) \| & \leq L_{f,x} \| x - x' \|,
\end{align}
\end{subequations}
hold for all $x, x' \in \mathcal{X}$ and $u, u' \in \mathcal{U}$.
Hence, \eqref{eq:lipschitz_f_u} implies that the bounded input approximation error induces a bounded additive disturbance $w$ on the system states,\noeqref{eq:lipschitz_f_x} % fixes a numbering issue do not remove please
\begin{align}
\begin{split}\label{eq:system_disturbed}
    x_{k+1} = f(x_k, u_k + e_k) \equiv f(x_k, u_k) + w_k,~\text{with} \\
    \| w_k \| = \| f(x_k, u_k + e_k) - f(x_k, u_k) \| \leq L_{f,u} \epsilon,
\end{split}
\end{align}
%exploiting knowledge of the approximation error bound $\epsilon > 0$.
\rev{exploiting knowledge of a bound $\epsilon > 0$ on the approximation error.} We emphasize that this approximation error bound is a design parameter. We will later show that it is indeed possible to construct a neural network approximator that satisfies any prescribed bound. 

Having rewritten the approximation error as a bounded additive disturbance in \eqref{eq:system_disturbed}, the closed-loop dynamics under the neural network controller fit the standard setting of robust MPC.

\subsection{Robust MPC Design and Closed-Loop Guarantees}\label{subsec:rmpc}

The general idea of robust MPC schemes is to analyze a nominal system model subject to appropriately tightened constraints, such that satisfaction of the tightened constraints by the undisturbed system implies satisfaction of the original constraints by the disturbed system.
This enables the establishment of robust \rev{recursive} feasibility guarantees in closed-loop operation.
Furthermore, the optimal value function of a robust MPC scheme is typically used as a Lyapunov function candidate to derive robust stability guarantees for the closed-loop system.

Without loss of generality, we adopt the robust MPC scheme proposed in \cite{marruedo2002}\footnote{\rev{This formulation uses open-loop disturbance propagation, so the resulting constraint tightening may increase with the prediction horizon and, in extreme cases, affect feasibility; this aspect is discussed in greater detail in Section~\ref{supp:sec:2dex}.}}, whose underlying optimal control problem is given by

\begin{subequations}\label{eq:ocp}
\begin{flalign} 
    \min_{\mathbf{u}_k} \quad & \sum_{n=0}^{N-1} l(\bar{x}_{n \mid k}, u_{n \mid k}) + V_\mathrm{T}(\bar{x}_{N \mid k}) \label{eq:ocp_cost} \\
    \text{s.~t.~} ~~ & \forall n \in \{0, \ldots, N-1 \}: \nonumber\\
                     & ~~~~ \bar{x}_{n+1 \mid k} = f(\bar{x}_{n \mid k}, u_{n \mid k}), ~ \bar{x}_{0 \mid k} = x_k, \label{eq:ocp_model} \\
                     & ~~~~ u_{n \mid k} \in \mathcal{U} \ominus \{u \in \mathbb{R}^{n_u} \mid \| u \| \leq \epsilon \}, \label{eq:ocp_input_constraints} \\
                     & ~~~~ \bar{x}_{n \mid k} \in \mathcal{X} \ominus \mathcal{W}_n, \label{eq:ocp_state_constraints} \\
                     & ~~~~ \bar{x}_{N \mid k} \in  \mathcal{X}_\mathrm{T} = \{ x \in \mathbb{R}^{n_x} \mid V_\mathrm{T}(x) \leq \alpha_v \}. \label{eq:ocp_terminal_constraint}
\end{flalign}
\end{subequations}

Here, $\bar{x}_{n \mid k}$ denotes the nominal state prediction $n$ steps ahead of the current time $k$, $u_{n \mid k}$ denotes the planned future control inputs $n$ steps ahead of the current time $k$, and \rev{$N \in \mathbb{N}$}, denotes the prediction horizon.
The nominal state predictions are computed using the nominal system model \eqref{eq:ocp_model} initialized with the most recent state information $x_k$.
The control objective is encoded in the cost function in \eqref{eq:ocp_cost}, which is composed of the stage cost $l: \mathbb{R}^{n_x} \times \mathbb{R}^{n_u} \to [0, \infty)$ and the terminal cost $V_\mathrm{T}: \mathbb{R}^{n_x} \to [0, \infty)$.
We assume the following for the stage and the terminal cost functions.
\begin{assum}\label{ass:cost_function}
    The stage  and the terminal cost functions are twice continuously differentiable, i.e., $l \in C^2(\mathbb{R}^{n_x} \times \mathbb{R}^{n_u}, [0, \infty))$ and $V_\mathrm{T} \in C^2(\mathbb{R}^{n_x}, [0, \infty))$.
    % Moreover, it holds that $l(0, 0) = 0$ and $V_\mathrm{T}(0) = 0$.
    Moreover, $l$ and $V_\mathrm{T}$ are positive definite \rev{w.r.t. to the origin}.
\end{assum}
Note that Assumption \ref{ass:cost_function} is satisfied by many commonly used stage and terminal cost functions, including quadratic costs typically employed in regulation (and, by extension, tracking) problems.
Moreover, Assumption \ref{ass:cost_function} implies that the stage and terminal cost functions are Lipschitz continuous on the compact domain $\mathcal{X} \times \mathcal{U}$ and $\mathcal{X}$, respectively, with Lipschitz constants $L_l, L_v > 0$.

We impose the input constraint in \eqref{eq:ocp_input_constraints}, where $\ominus$ denotes the Pontryagin difference. %\footnote{Given two sets $\mathcal{A}$ and $\mathcal{B}$, their Pontryagin difference is defined as $\mathcal{A} \ominus \mathcal{B} := \{ a \in \mathcal{A} \mid \forall b \in \mathcal{B}: a+b \in \mathcal{A} \}$.}
This tightening of the original input constraints in \eqref{eq:constraints} accounts for the approximation error introduced by the neural network controller.

Similarly, the state constraints in \eqref{eq:constraints} appear in a tightened form in \eqref{eq:ocp_state_constraints}, where they are imposed on the nominal system state.
This tightening captures the effect of the (propagated) approximation error of the neural network controller on the system's behavior.
It is performed based on the sets
\begin{align}
    \mathcal{W}_n \! = \! \Bigg \{ \! e \! \in \! \mathbb{R}^{n_x} \mid \| e \| \! \leq \! \frac{L_{f,x}^n \! - \! 1 }{L_{f,x} \! - \! 1} L_{f,u} \epsilon \! \Bigg \}. \label{eq:ocp_error_set}
\end{align}
Here the sets $\mathcal{W}_n$ overapproximate the possible --- but a priori unknown --- realizations of the mismatch between the nominal system and the disturbed system after $n$ steps.
% The sets $\mathcal{W}_n$, as specified in \eqref{eq:ocp_error_set}, overapproximate the error between the nominal system $x_{k+1} = f(x_k, u_k)$ and the disturbed system as defined in \eqref{eq:system_disturbed} after $n$ steps.
Hence, the constraint tightening in \eqref{eq:ocp_state_constraints} ensures that the state constraints will be satisfied by the system \eqref{eq:systemdesc} under the approximate neural network controller.

The terminal region constraint in \eqref{eq:ocp_terminal_constraint}, with a suitably chosen parameter $\alpha_v > 0$, is used to guarantee \rev{recursive} feasibility and stability of the closed-loop system under the MPC.

The optimal control problem \eqref{eq:ocp} is repeatedly solved at any time $k$ with updated initial condition $x_k$, and the first element $u_{0 \mid k}^*$ of the optimal input sequence $\mathbf{u}_k^* = (u_{0 \mid k}^*, \ldots, u_{N-1 \mid k}^*)$ is applied to the system~\eqref{eq:systemdesc}.
This implicitly defines the MPC feedback law
\begin{equation}\label{eq:mpcmap}
    u_k = \kappa(x_k),~ \kappa(x_k) := u_{0 \mid k}^*.
\end{equation}
We assume that the following conditions hold \cite{marruedo2002}.

\begin{assum}\label{ass:terminal}
Let $\alpha>0$ %and let $V_{\mathrm T}:\mathbb{R}^{n_x}\to\mathbb{R}_{\ge 0}$ be a terminal cost.
and define the set $\Phi := \{ x \in \mathcal{X} \mid V_{\mathrm T}(x) \le \alpha \}$.
% \begin{equation}
%     \Phi := \{ x \in \mathcal{X} \mid V_{\mathrm T}(x) \le \alpha \}.
% \end{equation}
Assume there exists a terminal feedback controller $u = \kappa_\mathrm{T}(x),$ $\kappa_{\mathrm T}:\Phi\to\mathcal U$ and a constant
$\alpha_v$ with $0<\alpha_v\le \alpha$ such that:
\begin{enumerate}

\item[(i)] The terminal controller satisfies the tightened input constraints when evaluated for $x \in \Phi$, i.e., \(\Phi \subset \bigl\{ x \in \mathcal{X} \bigm| \kappa_{\mathrm{T}}(x) \in \mathcal{U} \ominus \{ u \in \mathbb{R}^{n_u} \mid \|u\| \leq \epsilon \}\bigr\}\). \\

%\begin{equation}
%    \Phi \subset \bigl\{ x \in \mathcal{X} \bigm| 
%    \kappa_{\mathrm{T}}(x) \in 
%    \mathcal{U} \ominus \{ u \in \mathbb{R}^{n_u} \mid \|u\| \leq \epsilon \}\bigr\}.
%\end{equation}
% i.e., the terminal control action remains strictly inside the input constraints on $\Phi$.

\item[(ii)] The set $\Phi$ and the terminal controller are such that $f(x,\kappa_{\mathrm T}(x)) \in \mathcal{X}_{\mathrm T}$ %terminal constraint set used in the optimal control problem is defined by \eqref{eq:ocp_terminal_constraint} such that $\mathcal{X}_{\mathrm T} \subset \Phi$.
% Moreover,
% \begin{equation}
%     f(x,\kappa_{\mathrm T}(x)) \in \mathcal{X}_{\mathrm T}
%     % \quad \forall x \in \mathcal{X}_{\mathrm T},
% \end{equation}
holds for all $x \in \Phi$, where $\mathcal{X}_{\mathrm T} = \{ x \in \mathbb{R}^{n_x} \mid V_\mathrm{T}(x) \leq \alpha_v \}$ is the terminal region used in the optimal control problem in \eqref{eq:ocp_terminal_constraint}.
Consequently, $\Phi$ is robustly forward invariant for the system \eqref{eq:system_disturbed} under $u=\kappa_{\mathrm T}(x)$. \\

\item[(iii)] The terminal cost $V_{\mathrm T}$ is a Lyapunov function for the nominal system under the terminal controller on $\Phi$, i.e., there exist class-$\mathcal K_\infty$ functions $\underline{\alpha},\overline{\alpha}$ satisfying
\begin{align}
    & \qquad \underline{\alpha}(\|x\|)\leq V_{\mathrm T}(x)\leq \overline{\alpha}(\|x\|),~\text{and} \\
    & V_{\mathrm T}\bigl(f(x,\kappa_{\mathrm T}(x))\bigr) - V_{\mathrm T}(x) \leq -\, l\bigl(x,\kappa_{\mathrm T}(x)\bigr),
\end{align}
for all $x \in \Phi$. \\
% \item[(iv)] The neural network approximation error $\epsilon$ satisfies $L_{f,u}\epsilon \leq (\alpha-\alpha_v)(L_v\,\rev{L_{f,x}^{N-1}})^{-1}$.
\item[(iv)] The neural network approximation error $\epsilon$ satisfies $L_{f,u}\epsilon \leq \frac{\alpha-\alpha_v}{L_v\,\rev{L_{f,x}^{N-1}}}$.
% \begin{equation}
%     L_{f,u}\epsilon \leq \frac{\alpha-\alpha_v}{L_v\,\rev{L_{f,x}^{N-1}}}.
% \end{equation}
\end{enumerate}
\end{assum}

%\begin{assum}\label{ass:terminal_ctrl}
%    There exists a local controller $u = \kappa_\mathrm{T}(x)$, and the terminal cost function $V_\mathrm{T}$ is an associated Lyapunov function, such that $V_\mathrm{T}(f(x, \kappa_\mathrm{T}(x))) - V_\mathrm{T}(x) \leq - l(x, \kappa_\mathrm{T}(x))$ for all $x \in \Phi$, where $\Phi = \{ x \in \mathcal{X} \mid V_\mathrm{T}(x) \leq \alpha \}$, $\Phi \subset \{ x \in \mathcal{X} \mid \kappa_\mathrm{T}(x) \in \mathcal{U} \ominus \{ u \in \mathbb{R}^{n_u} \mid \| u \| \leq \epsilon \} \}$, for a constant $\alpha > 0$.
%\end{assum}

% \begin{assum}\label{ass:terminal_set}
%     The constant $\alpha_v > 0$, used to define the terminal region in \eqref{eq:ocp_terminal_constraint}, is such that $f(x, \kappa_\mathrm{T}(x)) \in \mathcal{X}_\mathrm{T}$ holds for all $x \in \Phi$.
% \end{assum}

Under these assumptions, the MPC defined by the optimal control problem \eqref{eq:ocp} is \rev{recursive} feasible and robustly stabilizes the system \eqref{eq:systemdesc}, as stated in the following theorem.

\begin{thm}\label{thm:mpc_feas_stab}(\rev{Recursive} feasibility and stability \cite{marruedo2002}.)
    Suppose Assumptions \ref{ass:system_dynamics}, \ref{ass:constraint_sets}, \ref{ass:cost_function}, and \ref{ass:terminal} hold. 
    Then, if the optimal control problem \eqref{eq:ocp} is initially feasible, it remains recursively feasible.
    Moreover, the origin is input-to-state stable for the disturbed system \eqref{eq:system_disturbed} under the MPC feedback law~\eqref{eq:mpcmap}.
\end{thm}

\begin{cor}\label{cor:invariance}
    The set of feasible initial conditions, defined as $\mathcal{X}_N = \{ x \in \mathcal{X} \mid \text{Problem}~ \eqref{eq:ocp}\text{~admits~ a~ solution} \}$, is robustly forward invariant for the disturbed system \eqref{eq:system_disturbed} under the MPC feedback law \eqref{eq:mpcmap}.
    That is, $x \in \mathcal{X}_N \implies (f(x, \kappa(x)) + w) \in \mathcal{X}_N$ holds for all disturbances $w$ satisfying $\| w \| \leq L_{f,u} \epsilon$.
\end{cor}

These results establish the foundation for extending the safety and stability guarantees of the MPC scheme to the neural network controller, as discussed in the next section.\footnote{\rev{For nonconvex MPC problems, the optimizer need not be unique. In this case, we consider a fixed solution branch, i.e., a consistent choice of optimizer as the state varies. Global optimality is not required. This convention applies throughout the subsequent analysis.}}

%%%%%%%%%%%%%%%%%%%%%%%%%%%%%%%%%%%%%%%%%%%%%%%%%%%%%%%%%%%%%%%%%%%%%%%%%%%%%%%%
\section{Certified Neural Approximation of the MPC Feedback Law}\label{sec:lmpcaeb}

We begin by reviewing the fundamentals of feed-forward neural networks, which we employ as function approximators to learn the MPC feedback law. 
Thereafter, we present our main contribution.
First, we derive conditions under which the neural network is guaranteed to satisfy a prescribed approximation error bound.
Second, we discuss how safety and stability guarantees of the original MPC transfer to the neural network controller.
The section concludes with an algorithm for constructing the neural network controller.

\subsection{Feed-Forward Neural Networks}\label{subsec:ffnn}

\rev{We employ neural networks to approximate the feedback law $\kappa_{\mathrm{NN}}$, motivated by their ability to represent nonlinear functions accurately \cite{hornik1991,hornik1993}, while retaining a fixed online evaluation structure that is independent of the size of the training dataset. This is particularly attractive for our approach, where large offline datasets are required, as discussed in Section~\ref{subsec:sig}.}
Specifically, we restrict ourselves to feed-forward neural networks trained offline via supervised learning techniques and ensure differentiability of the neural networks by choosing differentiable activation functions. 
Particularly, the $j^\text{th}$ layer of a NN can be written as
\begin{equation}
    h_j = \sigma_j \circ \zeta_j(h_{j-1}) := \sigma_j \left( A_{j} h_{j-1} + b_j \right). \label{eq:nnlayer}
\end{equation}
Here, $\zeta_j$ denotes the affine transformation of the $j^\text{th}$ layer's inputs, defined through the weight matrix $A_j \in \mathbb{R}^{n_{h_j} \times n_{h_{j-1}}}$ and the bias vector $b_j \in \mathbb{R}^{n_{h_j}}$. 
The nonlinear function $\sigma_j: \mathbb{R}^{n_{h_j}} \to \mathbb{R}^{n_{h_j}}$ encodes the element-wise application of an activation function to $\zeta_j(h_{j-1})$ to produce the output $h_j$.
A feed-forward NN is now the composition of all layers and their activations, denoted as
\begin{equation}
    \kappa_{\mathrm{NN}}(x):= \zeta_m \circ \sigma_{m-1}\circ \zeta_{m-1} \circ \hdots \circ \sigma_1 \circ \zeta_1(x), \label{eq:nndesign}
\end{equation}
where $m$ is the total number of layers in the network. Since the activation functions $\sigma_j$ are differentiable, and $\zeta_j$ is affine and therefore differentiable, their composition \eqref{eq:nndesign} is also differentiable.
This implies that the NN \eqref{eq:nndesign} is Lipschitz continuous on the compact domain $\mathcal{X}$.

%\footnote{The compactness of the state set $\mathcal{X}$ is guaranteed by Assumption~\ref{ass:constraint_sets}.}

By exploiting Lipschitz continuity of the neural network, we derive conditions under which the NN satisfies a prescribed error bound in the next section.
However, it is important to note that the results presented are not limited to feed-forward neural networks.
Rather, any neural network that is Lipschitz continuous on $\mathcal{X}$ can be used.
We formalize this requirement in the more general assumption below.

\begin{assum}\label{ass:nn_lipschitz}
    The neural network approximator $\kappa_\mathrm{NN}$ is Lipschitz continuous on $\mathcal{X}$, i.e., there exists a constant $0 < L_\mathrm{NN} < \infty$ such that $\| \kappa_\mathrm{NN}(x_1) - \kappa_\mathrm{NN}(x_2) \| \leq L_\mathrm{NN} \| x_1 - x_2 \|$ holds for all $x_1, x_2 \in \mathcal{X}$.
\end{assum}

%}

\subsection{Approximation Error Bounds}\label{subsec:sig}
%\todo{
%    \begin{itemize}
 %       \item Proposition LICQ, SOSC, SCS
 %       \item Sketch of proof via rewritting
 %       \item theorem about eta and delta
%        \item remark (maybe in next section)
%    \end{itemize}

The results presented here develop the Lipschitz-based approximation approach
introduced in \cite[Theorem 1]{alsmeier24} into a certified approximation setting for robust nonlinear MPC.
As in that work, our approach relies primarily on two key aspects: (i) the careful design of the training data and the accuracy of the NN approximator on the training data, and (ii) regularity properties of the MPC feedback law \eqref{eq:mpcmap}.
We begin by discussing the requirements for the training data and the neural network's accuracy on it.

The neural network is trained on a dataset $\mathcal{D} = \{ (x_i, \kappa(x_i)) \mid i \in \{1, \ldots, n_{\mathcal{D}} \}\}$ consisting of $n_\mathcal{D}$ observations of the MPC feedback law \eqref{eq:mpcmap} on the domain $\mathcal{X}$.
Specifically, we will explore coverage properties of the set of training inputs, $\{ x_i \mid i = 1, \ldots, n_\mathcal{D} \}$, on $\mathcal{X}$, as defined below.

%\begin{defn}\label{def:delta_cover}
%    Let $\mathcal{A}$ be a set in a metric space with metric \rev{$\rho$}, and let $0 < \delta < \infty$ be a constant. A collection $\mathcal{P} \subset \mathcal{A}$ of points is a $\delta$-coverage of $\mathcal{A}$ if, for every $a \in \mathcal{A}$, there exists $p \in \mathcal{P}$ such that $d(a, p) \leq \delta$.
%\end{defn}

\begin{defn}\label{def:delta_cover}
\rev{Let $(\mathcal A,\rho)$ be a metric space and let $\infty >\delta>0$.
A subset $\mathcal P\subseteq\mathcal A$ provides a $\delta$-coverage of $\mathcal A$ if, for every $a\in\mathcal A$, there exists $p\in\mathcal P$ such that $\rho(a,p)\leq\delta$.}
\end{defn}

\rev{We further assume that an arbitrary prescribed training accuracy on a dataset $\mathcal{D}$ can be attained through a suitable choice of network architecture and training procedure, as formalized in Assumption~\ref{ass:nn_training}. This assumption is motivated by Lemma~\ref{lem:nn_interpolates_finite}, which establishes the existence of a neural network achieving any prescribed accuracy on a finite dataset.}

\begin{lem}(Existence of an $\epsilon_{\mathcal D}$-interpolating NN; proof in supplementary Section~\ref{supp:sec:finite-dataset-approximation}.)\label{lem:nn_interpolates_finite}
Let $\mathcal{D} = \{ (x_i, y_i) \in \mathbb{R}^d \times \mathbb{R}^p \mid i \in \{1, \ldots, n_{\mathcal{D}} \}\}$ be a dataset \rev{with pairwise distinct inputs \(x_i\)} and $\epsilon_{\mathcal D}>0$ a given constant. Then there exists a feedforward neural network 
$\tilde f:\mathbb{R}^d\to\mathbb{R}^p$ with \rev{any} continuous non-polynomial activation function such that $\|\tilde f(x_i)-y_i\|\le \epsilon_{\mathcal D}$ for all $i=1,\dots,n_{\mathcal{D}}$.
\end{lem}

\begin{assum}\label{ass:nn_training}
Given $\epsilon_\mathcal{D} > 0$ and a training dataset $\mathcal{D}$, 
there exists a choice of training procedure and network architecture
such that the trained network $\kappa_\mathrm{NN}$ satisfies
$\| \kappa(x_i) - \kappa_\mathrm{NN}(x_i) \| \leq \epsilon_\mathcal{D}$
for all $(x_i,\kappa(x_i)) \in \mathcal{D}$.
\end{assum}

% \rev{\begin{rem}
% Assumption~\ref{ass:nn_training} states only that a prescribed in-sample accuracy on \(\mathcal D\) can be attained through suitable architectural and training choices. It neither guarantees convergence of a particular training algorithm nor implies generalization beyond the training data. While Lemma~\ref{lem:nn_interpolates_finite} establishes the existence of a neural network with the desired in-sample accuracy, the actually achieved error of a trained network can be evaluated directly as
% \[
% \hat{\epsilon}_\mathcal{D}
% :=
% \max_{(x_i,\kappa(x_i))\in\mathcal D}
% \left\|
% \kappa(x_i)-\kappa_{\mathrm{NN}}(x_i)
% \right\|.
% \]
% Hence, Assumption~\ref{ass:nn_training} can be verified by checking that
% \(\hat{\epsilon}_D\leq\epsilon_D\), where \(\epsilon_D\) is the training-error bound used in the subsequent analysis.
% \end{rem}}
\rev{Lemma~\ref{lem:nn_interpolates_finite} establishes the existence of a feedforward neural network with any prescribed in-sample accuracy, whereas Assumption~\ref{ass:nn_training} additionally requires that such accuracy be attainable through a suitable training procedure and architecture. Both statements concern only the samples in $\mathcal D$; they imply neither convergence of a particular training procedure nor generalization beyond the training data. The training-set bound is extended to a uniform approximation bound below by combining it with the regularity properties of the MPC and NN policies and the coverage of the training inputs.
We therefore begin by considering the regularity of the MPC feedback law \eqref{eq:mpcmap}.}
To this end, we note that the optimal control problem \eqref{eq:ocp} is a parametric optimization problem, where the parameter is the initial state $x_k$ in \eqref{eq:ocp_model}.
Specifically, we can express the nominal state $\bar{x}_{n \mid k},~ n \in \{0, \ldots, N \}$, in \eqref{eq:ocp_model} as a function of the initial state and the input (sub-)sequence, as $\bar{x}_{n \mid k} = \psi(n; x_k, \mathbf{u}_{k})$.
Since $\psi$ is defined via the $n$-fold application of the dynamics $f$, and $f$ is twice continuously differentiable in $x$ and $u$, $\psi$ is twice continuously differentiable in $x_k$ and $\mathbf{u}_k$.
Substituting $\psi(n; x_k, \mathbf{u}_k)$ for the states $\bar{x}_{n \mid k}$ in the cost function  \eqref{eq:ocp_cost} and the constraints \eqref{eq:ocp_state_constraints} and \eqref{eq:ocp_terminal_constraint} enables us to write the optimal control problem \eqref{eq:ocp} generically as
\begin{align}
\begin{split}\label{eq:ocp_param_program}
    \min_{\mathbf{u}_k} \quad & J(\mathbf{u}_k, x_k) \\
    \text{s.~t.~} ~~ & G_i(\mathbf{u}_k, x_k) \leq 0,~ i = 1, \ldots, N(n_c^u \! + \! n_c^x) \! + \! 1.
\end{split}
\end{align}
Again, since self-compositions and sums of twice continuously differentiable functions remain twice continuously differentiable, Assumptions \ref{ass:system_dynamics}, \ref{ass:constraint_sets}, and \ref{ass:cost_function} imply that the functions $J$ and $G_i,~ i = 1, \ldots, N(n_c^u  + n_c^x) + 1$ in \eqref{eq:ocp_param_program}, are twice continuously differentiable on the interior of $\mathcal{U}^N \times \mathcal{X}$.
We will explore the following result from the field of parametric optimization.

\begin{prop}(Differentiability of optimal solutions \cite{Bueskens2001}.)\label{prop:diff_opt_solution}
    Given $x_k$, suppose that $\mathbf{u}_k^*$ is an admissible point for the parametric optimization problem \eqref{eq:ocp_param_program} which satisfies the Karush-Kuhn-Tucker conditions with Lagrange multiplier vector $\bm{\mu}_k^*$.
    The corresponding active set is given as the index set $\mathcal{I}_\mathrm{a} = \{i \in \{ 1, \ldots, N(n_c^u + n_c^x) + 1 \} \mid G_i( \mathbf{u}_k^*, x_k) = 0 \}$.
    Furthermore, suppose that
    \begin{enumerate}
        \item[(i)] the linear independence constraint qualification (LICQ) is satisfied at $(\mathbf{u}_k^*, x_k)$, i.e., the gradients $\nabla_{\mathbf{u}_k} G_i(\mathbf{u}^*_k, x_k),~ i \in \mathcal{I}_\mathrm{a}$, are linearly independent,
        \item[(ii)] strict complementary slackness (SCS) holds, i.e., $\mu_{k,i}^* > 0$ for all $i \in \mathcal{I}_\mathrm{a}$, and
        \item[(iii)] the second-order sufficient condition (SOSC) is satisfied at $(\mathbf{u}_k^*, x_k)$, i.e., the Hessian of the Lagrangian at $(\mathbf{u}_k^*, x_k)$ is positive definite when restricted to the tangent space corresponding to the active constraints.
    \end{enumerate}
    Then, there exist a neighborhood $\mathcal{N}(x_k) \subset \mathcal{X}$ of $x_k$ and continuously differentiable functions $\mathbf{u}: \mathcal{N}(x_k) \to \mathcal{U}^N$ and $\bm{\mu}: \mathcal{N}(x_k) \to \mathbb{R}^{(N-1)(n_c^u + n_c^x)+1}$ such that
    \begin{enumerate}
        \item[(a)] $\mathbf{u}(x_k) = \mathbf{u}_k^*$, $\bm{\mu}(x_k) = \mu_k^*$,
        \item[(b)] the active set is constant in $\mathcal{N}(x_k)$,
        \item[(c)] LICQ is satisfied for all $z \in \mathcal{N}(x_k)$, and
        \item[(d)] for all $z \in \mathcal{N}(x_k)$, the pair $(\mathbf{u}(x_k), \bm{\mu}(x_k))$ is a strong local minimum of the problem \eqref{eq:ocp_param_program} with initial state $z$, i.e., LICQ, SCS, and SOSC are satisfied.
    \end{enumerate}
\end{prop}

\begin{cor}
    Under the same conditions as in Proposition \ref{prop:diff_opt_solution}, the MPC feedback law \eqref{eq:mpcmap} is locally Lipschitz continuous on the neighborhood $\mathcal{N}(x_k)$ of $x_k$, i.e., there exists a constant $0 < L_\mathrm{MPC}(x_k) < \infty$ such that $\| \kappa(x_1) - \kappa(x_2) \| \leq L_\mathrm{MPC}(x_k) \| x_1 - x_2 \|$ holds for all $x_1, x_2 \in \mathcal{N}(x_k)$.
\end{cor}
We now derive an approximation bound linking dataset coverage and Lipschitz constants.

\begin{thm}\label{thm:deltaeta}(Approximation Error Bound.)
    Suppose Assumptions~\ref{ass:system_dynamics}, \ref{ass:constraint_sets}, \ref{ass:cost_function}, \ref{ass:terminal}, \ref{ass:nn_lipschitz}, and \ref{ass:nn_training} hold.
    Furthermore, suppose that the current state $x_k$ of the system \eqref{eq:systemdesc} is such that the conditions in Proposition \ref{prop:diff_opt_solution} are satisfied.
    Then, for a given training dataset $\mathcal{D}$, given $\epsilon_\mathcal{D} > 0$, and for \rev{a prescribed} $\epsilon > \epsilon_\mathcal{D}$, the approximation error is bounded by 
    \begin{align}\label{eq:error_bound}
        \lVert \kappa(x_k) - \kappa_\mathrm{NN}(x_k) \rVert \leq \epsilon
    \end{align}
    if there exists a training data point $(x^*, \kappa(x^*)) \in \mathcal{D}$ that satisfies $\lVert x_k - x^* \rVert  \leq (L_\mathrm{MPC}(x_k) + L_\mathrm{NN})^{-1} (\epsilon - \epsilon_\mathcal{D})$ and $x^* \in  \mathcal{N}(x_k)$. 
\end{thm}

\begin{pf}
    We begin by writing the left-hand side of \eqref{eq:error_bound} as
    \begin{align*}
        \lVert  \kappa(x_k) - \kappa_\mathrm{NN}(x_k) \rVert =~ & \lVert \kappa(x_k) - \kappa(x^*) + \kappa_\mathrm{NN}(x^*) \\ 
        & - \kappa_\mathrm{NN}(x_k) + \kappa(x^*) - \kappa_\mathrm{NN}(x^*) \rVert.
    \end{align*}
    Applying the triangle inequality, we obtain the upper bound
    \begin{align*}
        \lVert \kappa(x_k) - \kappa_\mathrm{NN}(x_k) \rVert \leq~ & \lVert \kappa(x_k) - \kappa(x^*) \rVert \\
        &+ \lVert \kappa_\mathrm{NN}(x^*) - \kappa_\mathrm{NN}(x_k) \rVert \\
        & + \lVert \kappa(x^*) - \kappa_\mathrm{NN}(x^*) \rVert.
    \end{align*}
    By Assumption \ref{ass:nn_lipschitz}, and as $x^* \! \in \! \mathcal{N}(x_k)$, we obtain further that
    \begin{align*}
        \lVert \kappa_\mathrm{NN}(x^*) - \kappa_\mathrm{NN}(x_k) \rVert & \leq L_\mathrm{NN} \lVert x^* - x_k \rVert, \\
        \lVert \kappa(x_k) - \kappa(x^*) \rVert & \leq L_\mathrm{MPC}(x_k) \lVert x_k - x^* \rVert.
    \end{align*}
    Now, it follows straightforwardly that
    \begin{align*}
        \lVert \kappa(x_k) - \kappa_\mathrm{NN}(x_k) \rVert \leq~ & (L_\mathrm{MPC}(x_k) + L_\mathrm{NN}) \lVert x_k - x^* \rVert \\
        & + \epsilon_\mathcal{D},
    \end{align*}
    which, together with the requirement that $\lVert x_k - x^* \rVert \leq (L_\mathrm{MPC}(x_k) + L_\mathrm{NN})^{-1} (\epsilon - \epsilon_\mathcal{D})$, implies the result in \eqref{eq:error_bound}.
\end{pf}

We now turn to a detailed discussion of the conditions in Proposition~\ref{prop:diff_opt_solution}.%, which remain critical for the validity of our main result.

\begin{rem}\label{rem:licq_scs_sosc}
Assuming that LICQ, SOSC, and SCS hold can be restrictive. 
However, these assumptions are often employed in explicit MPC and parametric optimization to ensure that the optimal solution is well-behaved, see, e.g., \cite{bemporad2002} or \cite{zanon2021}. 
LICQ is typically ensured by design of the constraints \eqref{eq:ocp_input_constraints}, \eqref{eq:ocp_state_constraints}, and \eqref{eq:ocp_terminal_constraint}, or via regularization or constraint relaxation techniques, see, e.g., \cite{rawlings1999,rawlings2017}. 
SCS is usually only lost when arbitrarily small perturbations of the state $x_k$ lead to active set changes.
However, the sets on which this is the case are typically lower-dimensional subsets of $\mathcal{X}$, and hence, Lebesgue null sets.
As a result, it is unlikely in practice that a state sample lies in these sets. Satisfaction of the SOSC is more challenging; however, our approach allows for violations of the SOSC on Lebesgue-null sets, analogous to the SCS condition. For a more comprehensive overview of the regularity properties of parametric optimization problems, see, e.g., \cite{fiacco1990}.
\end{rem}

To better understand the implications of Theorem \ref{thm:deltaeta}, we proceed by discussing some of its practical aspects.
\begin{rem}\label{rem:eta_epsilon}
If an upper bound $L_\mathrm{MPC}^\mathrm{max}$ on the local Lipschitz constants $L_\mathrm{MPC}(x_k)$ can be estimated over all regions of the state space where the MPC feedback law is locally Lipschitz continuous, Theorem~\ref{thm:deltaeta} generalizes accordingly.
Specifically, a state-independent distance bound $\delta := (L_\mathrm{MPC}^\mathrm{max} + L_\mathrm{NN})^{-1} (\epsilon - \epsilon_\mathcal{D})$ can then be used to establish the result\rev{, at the cost of potentially increased conservatism}.
Hence, Theorem \ref{thm:deltaeta} provides a condition on the minimum coverage of $\mathcal{X}$ that must be achieved by the training inputs.
\rev{To enforce the distance condition in Theorem~\ref{thm:deltaeta}, one can choose a parameter $0<\eta\leq\delta$ and construct the training dataset to provide an $\eta$-coverage of $\mathcal{X}$. However, $\eta$-coverage alone does not ensure that the corresponding training point lies in the same neighborhood of local Lipschitz continuity as $x_k$. 
These regions are determined by the underlying MPC problem and may have arbitrary geometry; in particular, they need not be norm balls or have a uniform radius. For each state $x_k$, the relevant neighborhood is the one associated with the corresponding region of the MPC problem. The result therefore remains conditional on the existence, for each considered state $x_k$, of a training point that satisfies both the distance and neighborhood conditions of Theorem~\ref{thm:deltaeta}.}
%To ensure that, for any $x_k \in \mathcal{X}$ that satisfies the regularity conditions in Proposition \ref{prop:diff_opt_solution}, a suitable training data point $(x^*, \kappa(x^*))$ exists, one can choose a parameter $0 < \eta \leq \delta$, and construct the training dataset so that the training inputs provide an $\eta$-coverage of $\mathcal{X}$.
The parameter $\eta$ thus acts as a tuning parameter, which must be selected heuristically or via appropriate tuning strategies tailored to the specific application.
%Note that the generalization of Theorem \ref{thm:deltaeta} to the entire state space $\mathcal{X}$ is, in a sense, conservative, as it is practically sufficient to establish the error bound --- and, consequently, the desired feasibility and stability properties under the NN controller --- only for the set of states reachable in closed-loop for the given task.
\end{rem}

\subsection{\rev{Recursive} feasibility and Stability Guarantees}
\label{subsec:repfeas_and_stab}

Building on the results of Theorem \ref{thm:deltaeta}, we can now guarantee safety, i.e., constraint satisfaction, and stability of the system \eqref{eq:systemdesc} under the neural network controller.
Specifically, Theorem \ref{thm:deltaeta} provides conditions under which a neural network satisfies a prescribed bound on the approximation error, which can be used to design the underlying MPC to be robust against such errors, as described in Section \ref{sec:problem}.
This robustness enables safe substitution of the MPC with the neural network controller.
In particular, if the conditions of Theorem \ref{thm:deltaeta} are satisfied (point-wise) along the closed-loop state trajectory of system \eqref{eq:systemdesc} under the neural network controller, the safety and stability guarantees established in Theorem \ref{thm:mpc_feas_stab} are preserved.
Moreover, under these conditions, the neural network controller maintains the invariance of the feasible set of initial conditions for the closed-loop system, as established in Corollary \ref{cor:invariance}.

The following theorem provides a formal statement of the result under practical and realistically attainable conditions, explicitly accounting for stochastic disturbances that arise in closed-loop operation, which results in a probabilistic view on the results in Proposition~\ref{prop:diff_opt_solution} and Theorem~\ref{thm:deltaeta}.

\begin{thm}\label{thm:closed_loop_guarantees}
    Suppose the following conditions hold:
    \begin{enumerate}
    
        \item[(i)] Assumptions~\ref{ass:system_dynamics}, \ref{ass:constraint_sets}, \ref{ass:cost_function}, \ref{ass:terminal}, \ref{ass:nn_lipschitz}, and \ref{ass:nn_training} hold.

        \item[(ii)] The system \eqref{eq:systemdesc} is subject to additive, bounded, independent, and identically distributed noise $\nu_k \in \mathcal{W_\nu} \subset \mathbb{R}^{n_x}$, where 
        \begin{enumerate}
            \item[(a)] $\mathcal{W}_\nu$ is compact and satisfies $0 \in \mathrm{int}(\mathcal{W}_\nu)$, and
            \item[(b)] $\nu_k$ follows a continuous probability measure $\mu_\nu$ that is absolutely continuous with respect to the Lebesgue measure $\lambda$, i.e., $\mu_\nu \ll \lambda$.
        \end{enumerate}

        \item[(iii)] The MPC problem \eqref{eq:ocp} is designed to be robust with respect to the additive, bounded, total disturbance $\bar{w}_k := w_k + \nu_k$, where $w_k$ is defined as in \eqref{eq:system_disturbed}.

        \item[(iv)] The set $\mathcal{Z}:= \{ x \in \mathcal{X}_N \mid ~\text{LICQ, SCS, or SOSC is vio-} \allowbreak \text{lated at}~ (\mathbf{u}^*, x) \}$, where $\mathbf{u}^*$ is the optimal solution of \eqref{eq:ocp} with initial condition $x$, satisfies $\lambda(\mathcal{Z})=0$, i.e., the set of states where regularity fails is a null set with respect to the Lebesgue measure.

        \item[(v)] The initial state satisfies $x_0 \in \mathcal{X}_N \setminus \mathcal{Z}$, i.e., the MPC problem is initially feasible, and the feedback law is regular at $x_0$.

        \item[(vi)] For each $x \in \mathcal{X}_N \setminus \mathcal{Z}$, there exists a training data point $(x^*, \kappa(x^*)) \in \mathcal{D}$ that satisfies $\| x - x^* \| \leq (L_\mathrm{MPC}^\mathrm{max} + L_\mathrm{NN})^{-1}(\epsilon - \epsilon_\mathcal{D})$ \rev{and $x^* \in  \mathcal{N}(x)$}.
        
    \end{enumerate}
    Then, the neural network controller 
    \begin{enumerate}
        \item[(1)] remains feasible along the closed-loop trajectory almost surely, and
        \item[(2)] renders the closed-loop system almost surely input-to-state stable with respect to the origin and the disturbance $\bar{w}_k$.
    \end{enumerate}
\end{thm}

\begin{pf}
    The proof proceeds in three main steps.
    First, we revisit the robust MPC design under the presence of an additional disturbance, as specified in condition (iia), and demonstrate the feasibility of the robust MPC formulation required by condition (iii).
    Second, we show that the feedback law of this robust MPC is almost surely sufficiently regular at the sampled closed-loop states $x_k$ where it is evaluated, leveraging conditions (iia) and (iv).
    This guarantees the local Lipschitz continuity of the feedback law in neighborhoods around the observed closed-loop states.
    Third, leveraging conditions (i) and (v), we establish that the approximation error at any observed closed-loop state satisfies the bound stated in Theorem \ref{thm:deltaeta}.
    Together, these results imply that the guarantees provided by the robust MPC design in Theorem \ref{thm:mpc_feas_stab} transfer to the neural network controller.

    \textbf{Step 1:} We begin by considering the system \eqref{eq:systemdesc} subject to the additive disturbance $\nu_k$, given by $x_{k+1} = f(x_k, u_k) + \nu_k$.
    %\begin{align*}
    %    x_{k+1} = f(x_k, u_k) + \nu_k.
    %\end{align*}
    
    Analogously to the formulation in \eqref{eq:system_disturbed}, we define a single additive disturbance term $\bar{w}_k$ that captures both the effect of the approximation error induced by the neural network controller and the state disturbance $\nu_k$.
    Specifically,
    \begin{align*}
        x_{k+1} & = f(x_k, u_k+e_k) + \nu_k \equiv f(x_k, u_k) + \bar{w}_k,~\text{with} \\
        \bar{w}_k & = f(x_k, u_k+e_k) - f(x_k, u_k) + \nu_k.
    \end{align*}
    This yields the bound \(\| \bar{w}_k \| \! \leq \! \| f(x_k, u_k \! + \! e_k) \! - \! f(x_k, u_k) \| \! + \! \| \nu_k \|  \leq L_{f,u} \epsilon + \epsilon_\nu\), where $\epsilon_\nu > 0$ is a constant such that $\mathcal{W}_\nu \subseteq \{ \nu \mid \| \nu \| \leq \epsilon_\nu \}$.
        %\begin{align*}
    %    \| \bar{w}_k \| \! \leq \! \| f(x_k, u_k \! + \! e_k) \! - \! f(x_k, u_k) \| \! + \! \| \nu_k \|  \leq L_{f,u} \epsilon + \epsilon_\nu,
    %\end{align*
    Such a constant $\epsilon_\nu > 0$ exists due to the compactness of $\mathcal{W}_\nu$.
    Including \rev{$\epsilon_\nu$} in \rev{\eqref{eq:ocp_error_set}} yields
    \begin{align*}
        \mathcal{W}_n \! = \! \Bigg \{ \! e \! \in \! \mathbb{R}^{n_x} \mid \| e \| \! \leq \! \frac{L_{f,x}^n \! - \! 1 }{L_{f,x} \! - \! 1} \left( L_{f,u} \epsilon + \epsilon_\nu \right) \! \Bigg \}. \\
    \end{align*}
    This establishes the feasibility of the robust MPC design required in condition (iii), which is henceforth assumed.

    \textbf{Step 2:} Given condition (iia), for a fixed state $x_k \in \mathcal{X}_N \setminus \mathcal{Z}$, the successor state $x_{k+1}$ is a random variable determined by the disturbance $\nu_k$ and follows the probability measure $\mu_x$.
    Specifically, $\mu_x$ is the pushforward of $\mu_\nu$ under the translation map $T(z) = z + f(x_k,u_k+e_k)$.
    Since the Lebesgue measure $\lambda$ is translation invariant and $\mu_\nu \ll \lambda$ by assumption, it follows that $\mu_x \ll \lambda$, i.e., $\mu_x$ is absolutely continuous with respect to the Lebesgue measure.
    Together with condition (iv), this yields $\mu_x(\mathcal{Z}) = 0$, or equivalently, $\mathrm{Pr}[x_{k+1} \in \mathcal{X}_N \setminus \mathcal{Z}] = 1$ almost surely.
    By induction, and using condition (v), it follows that $x_k \in \mathcal{X}_N \setminus \mathcal{Z}$ for all $k$, almost surely.
    Consequently, the regularity conditions of Proposition \ref{prop:diff_opt_solution} are almost surely satisfied at all visited closed-loop states, implying local Lipschitz continuity of the MPC feedback law in neighborhoods of these states.

    \textbf{Step 3:} The established (local) Lipschitz continuity of the MPC feedback law in neighborhoods of all visited closed-loop states, together with conditions (i) and (vi), directly implies that the error at these states almost surely satisfies the bound stated in Theorem \ref{thm:deltaeta}.
    Since, by condition (iii), the underlying MPC is robust with respect to both the approximation error induced by the neural network controller and the additive disturbance $\nu_k$, the closed-loop guarantees established in Theorem \ref{thm:mpc_feas_stab} almost surely carry over to the neural network controller.
\end{pf}
\rev{We conclude the theoretical analysis by briefly discussing the conditions of Theorem~\ref{thm:closed_loop_guarantees}.
Condition (i) requires Assumptions~\ref{ass:system_dynamics}, \ref{ass:constraint_sets}, \ref{ass:cost_function}, \ref{ass:terminal}, \ref{ass:nn_lipschitz}, and \ref{ass:nn_training}, whose implications and restrictiveness were discussed upon introduction.}

\rev{Condition (ii) assumes bounded additive disturbances with a continuous distribution of compact support, encompassing, e.g., model-plant mismatch, external disturbances, and measurement noise. Condition (iia) is standard in robust and stochastic MPC, while (iib) covers, among others, uniform and truncated Gaussian distributions; the distribution itself need not be known explicitly or analytically tractable.}

\rev{Condition (iii) requires robustness of the MPC against the combined disturbance, achievable through standard constraint tightening as described in Section~\ref{subsec:rmpc}.}

\rev{Condition (iv) requires the irregularity set $\mathcal{Z}$ to have Lebesgue measure zero, supported by Remark~\ref{rem:licq_scs_sosc} and empirical evidence reported in ~\cite{zanon2021}.}

\rev{Finally, condition (v) requires initial MPC feasibility and excludes states where the feedback law lacks regularity.}

\subsection{Procedure for Certified NN Approximation}
\label{sec:pipline}

We now present a procedure to generate a NN controller that provides the guarantees discussed in Section~\ref{subsec:repfeas_and_stab} \rev{and discuss the limitations of implementing this method in a practical setting.}

A possible implementation is summarized in Algorithm~\ref{alg:pipeline}. The resulting NN is fixed at run time and does not require any online adaptation or re-training. \rev{In particular, once provisional bounds on the approximation error and the Lipschitz constants of the MPC law and the NN have been selected, sampling and training may be carried out in a single offline procedure, provided that the resulting network satisfies the prescribed bounds.} Thus, enlarging the training dataset does not affect online execution times, as the evaluation time at deployment depends only on the NN architecture.

We first specify a bound $\epsilon_\mathcal{D}$ on the expected error at the training points of the NN with respect to the underlying MPC law $\kappa(x)$. This bound may be chosen based on domain knowledge of the problem, an engineering guess, or a single preliminary NN training on an arbitrary initial dataset $\mathcal{D}_0$ from a similar domain. Given such a bound, we select the design parameter $\epsilon$ in Theorem~\ref{thm:deltaeta} so that $\epsilon \geq \epsilon_\mathcal{D}$, which can be chosen larger to include additional robustness margins.

Once $\epsilon$ has been fixed, we design the robust MPC as in Section~\ref{subsec:rmpc}. At this stage, state disturbances $\epsilon_\nu$ can also be incorporated. This yields a robust MPC feedback law $\kappa(x)$.

%Given the MPC design, we next provide bounds on the local Lipschitz constants of both the NN and the MPC on $\mathcal{X}$. For the MPC, we obtain an upper bound $L_\mathrm{MPC}$ on the local Lipschitz constants of the feedback law~\eqref{eq:mpcmap}, for instance via power iteration, finite-difference schemes, or specialized methods (see, e.g., \cite{teichrib2023,pauli2021}). Similarly, we fix a bound $L_\mathrm{NN}$ on the Lipschitz constant of the NN architecture to be used for approximation. In parallel, we estimate the set of initial conditions $\mathcal{X}_N$ for the robustly forward invariant set for the closed-loop system, since the closed-loop MPC law~\eqref{eq:mpcmap} is only defined on $\mathcal{X}_N$.% and training data must be drawn from this set.
Given the MPC design, we next provide bounds on the local Lipschitz constants of both the NN and the MPC on $\mathcal{X}$.
\rev{An upper bound $L_\mathrm{MPC}$ for the feedback law~\eqref{eq:mpcmap} may be obtained via power iteration, finite-difference schemes, or specialized methods~\cite{teichrib2023,pauli2021}. Similarly, we fix $L_\mathrm{NN}$ for the NN architecture and estimate the robustly
forward-invariant set $\mathcal{X}_N$, on which the MPC law is defined and from which the training data are drawn.}

With the bounds $\epsilon_\mathcal{D}$, $L_\mathrm{NN}$, and $L_\mathrm{MPC}$ specified and Theorem~\ref{thm:deltaeta}, we compute the coverage radius $\delta$. We then select a radius $\eta \leq \delta$ for the over-coverage. Having fixed $\eta$, we employ sub-sampling or space-filling methods \cite{kiss2024} to compute an $\eta$-coverage of $\mathcal{X}_N$, which is also a $\delta$-coverage of $\mathcal{X}_N$. Evaluating the MPC feedback law~\eqref{eq:mpcmap} on these points in open-loop yields the final training dataset $\mathcal{D}$.

%Finally, we initialize the neural network with the chosen architecture and train it once on the fixed dataset $\mathcal{D}$. \rev{The training objective is to obtain a NN controller $\kappa_\mathrm{NN}$ whose realized in-sample error $\hat \epsilon_\mathcal{D}$ satisfies $\hat{\epsilon}_{\mathcal D} := \max_{(x_i,\kappa(x_i))\in\mathcal D} \left\| \kappa(x_i)-\kappa_{\mathrm{NN}}(x_i) \right\| \leq \epsilon_{\mathcal D}$ and whose realized  $\hat L_\mathrm{NN}$ satisfies $\hat L_\mathrm{NN} \leq L_\mathrm{NN}$. Since both the realized in-sample error and Lipschitz constant are generally known only after training, additional hyperparameter tuning or retraining may be required if the prescribed bounds are not met. Any network that meets these two design inequalities, however, is admissible for our analysis, and no further iterations or refinements are required.}

\rev{We then initialize the neural network with the chosen architecture
and train it once on the fixed dataset $\mathcal{D}$. The trained network must satisfy
$\hat{\epsilon}_{\mathcal D}
:= \max_{(x_i,\kappa(x_i))\in\mathcal D}
\|\kappa(x_i)-\kappa_{\mathrm{NN}}(x_i)\|
\leq \epsilon_{\mathcal D}$
and $\hat L_\mathrm{NN}\leq L_\mathrm{NN}$.
Because these quantities are known only after training, hyperparameter
tuning or retraining may be required. Any network satisfying both
inequalities is admissible, and no further refinement is needed.}

\begin{algorithm}[!ht]
\caption{Generation of NN Controller}
\label{alg:pipeline}
\begin{algorithmic}[1]
\Require Constraint set $\mathcal{X}$.

\State \textbf{Specify or estimate} a nominal approximation error bound $\epsilon_\mathcal{D}$ 

\State \textbf{Choose} a design approximation error bound $\epsilon \geq \epsilon_\mathcal{D}$.

\State \textbf{Design} a robust MPC using $\epsilon$ and, optionally, a disturbance bound $\epsilon_\nu$, obtaining the feedback law $\kappa(x)$.

\State \textbf{Specify} Lipschitz bounds $L_{\mathrm{MPC}}$ and $L_{\mathrm{NN}}$ on $\mathcal{X}$.

\State \textbf{Estimate} the robustly forward invariant set $\mathcal{X}_N$ of the robust MPC.

\State \textbf{Compute} the required coverage radius $ \delta = \frac{\epsilon - \epsilon_\mathcal{D}}{L_{\mathrm{MPC}} + L_{\mathrm{NN}}}$ according to Theorem~\ref{thm:deltaeta}.

\State \textbf{Select} a radius $\eta \leq \delta$ for over-coverage.

\State \textbf{Compute} an $\eta$-coverage of $\mathcal{X}_N$ and form the training dataset $\mathcal{D}$ by evaluating $\kappa(x)$ on the coverage points.

\State \textbf{Train} \rev{the neural network $\kappa_\mathrm{NN}$ on $\mathcal{D}$, such that it satisfies $\hat \epsilon_\mathcal{D} \leq\epsilon_\mathcal{D} \leq \epsilon$ and $\hat L_\mathrm{NN} \leq L_\mathrm{NN}$}.

\State \textbf{Output} the final NN controller $\kappa_\mathrm{NN}(x)$, which approximates $\kappa(x)$ on $\mathcal{X}_N$ with error bound $\epsilon$.
\end{algorithmic}
\end{algorithm}
%\rev{
%\begin{rem}
%In practice, coverage-based certification is subject to the curse of dimensionality. For a bounded \(n_x\)-dimensional domain with characteristic width \(R\), the number of samples required to achieve a coverage radius \(\delta\) scales approximately as
%$
%n_{\mathcal D}
%=
%\mathcal{O}\!\left(
%\left(\delta^{-1} R \right)^{n_x}
%\right)
%$.
%Thus, maintaining a fixed coverage radius requires rapidly increasing offline sampling and training effort as the state dimension grows, while the online evaluation time remains determined by the NN architecture. Additionally, scalability depends not only on $n_x$, but also on the required radius $\delta$, and therefore on the acceptable approximation error margin and the Lipschitz bounds of the MPC and NN.
%\end{rem}
%}
%\vspace{-12pt}
\rev{
\begin{rem}
In practice, coverage-based certification suffers from the curse of dimensionality: for a bounded $n_x$-dimensional domain with characteristic width $R$, the required number of samples scales as $n_{\mathcal D}=\mathcal{O}((R/\delta)^{n_x})$.
Thus, maintaining a fixed coverage radius causes an increasing offline sampling and training effort with $n_x$, whereas online evaluation remains determined by the NN architecture. Scalability also depends on $\delta$, and hence on the acceptable approximation-error margin and the Lipschitz bounds.
\end{rem}
}
%%%%%%%%%%%%%%%%%%%%%%%%%%%%%%%%%%%%%%%%%%%%%%%%%%%%%%%%%%%%%%%%%%%%%%%%%%%%%%%%
\section{Example}\label{sec:ex} 

\rev{We evaluate the approach on an isothermal continuously stirred tank reactor in which the reaction \(2A+B\rightarrow C\), \(C\rightarrow 2D\), and the selective removal of \(D\) are realized. The continuous-time dynamics are
\begin{align}
    \dot c_A
    &= d\left(c_{A,f}-c_A\right)
       -2k_1 c_A^2c_B, \\
    \dot c_B
   & =  d\left(c_{B,f}-c_B\right)
       -k_1 c_A^2c_B, \\
    \dot c_C
    &= - d c_C
       +k_1 c_A^2c_B-k_2c_C, \\
    \dot c_D
    &= - d c_D
       +2k_2c_C-k_{\mathrm{sep}}c_D .
\end{align}
Here, $x=[c_A,c_B,c_C,c_D]$ denotes the species concentrations in
$\SI{}{\mol\per\liter}$, and $u=c_{A,f}$ is the controlled feed
concentration of species $A$ in $\SI{}{\mol\per\liter}$.
The dilution rate is $d=\SI{0.1}{\per\minute}$, and the fixed feed concentration is $c_{B,f}=\SI{0.75}{\mol\per\liter}$. The reaction parameters are
$k_1=\SI{3.2}{L^2\per\minute^2}$
$k_2=\SI{0.1}{\per\minute}$ and
$k_{\mathrm{sep}}=\SI{0.1}{\per\minute}$.
The model is discretized using fourth-order Runge--Kutta with
$T_s=\SI{0.5}{\minute}$.}
%I took this out as people should know an RK4. We can add it if their is space.

%Consequently, the disturbed discrete-time model is
%\begin{align}
% x_{1,k+1} &= x_{1,k}
% +\frac{T_s}{6}\left(\kappa_{1,k}^{(1)}
% +2\kappa_{2,k}^{(1)}+2\kappa_{3,k}^{(1)}
% +\kappa_{4,k}^{(1)}\right)+\tilde\nu_{1,k},\\
% x_{2,k+1} &= x_{2,k}
% +\frac{T_s}{6}\left(\kappa_{1,k}^{(2)}
% +2\kappa_{2,k}^{(2)}+2\kappa_{3,k}^{(2)}
% +\kappa_{4,k}^{(2)}\right)+\tilde\nu_{2,k},\\
% x_{3,k+1} &= x_{3,k}
% +\frac{T_s}{6}\left(\kappa_{1,k}^{(3)}
% +2\kappa_{2,k}^{(3)}+2\kappa_{3,k}^{(3)}
% +\kappa_{4,k}^{(3)}\right)+\tilde\nu_{3,k},\\
% x_{4,k+1} &= x_{4,k}
% +\frac{T_s}{6}\left(\kappa_{1,k}^{(4)}
% +2\kappa_{2,k}^{(4)}+2\kappa_{3,k}^{(4)}
% +\kappa_{4,k}^{(4)}\right)+\tilde\nu_{4,k}.
%\end{align}}
%\rev{with
%\begin{align}
%    \kappa_{1,k} &= f(x_k,u_k), \\
%    \kappa_{2,k} &= f\left(x_k+\frac{T_s}{2}\kappa_{1,k},u_k\right), \\
%    \kappa_{3,k} &= f\left(x_k+\frac{T_s}{2}\kappa_{2,k},u_k\right), \\
%    \kappa_{4,k} &= f\left(x_k+T_s\kappa_{3,k},u_k\right).
%\end{align}}

\rev{Over $\mathcal{X}\times\mathcal{U}$, the numerically estimated Lipschitz constants of the undisturbed discrete dynamics w.r.t. the $2$-norm are $L_{f,x}=1.11$ and $L_{f,u}=0.05$. We add $\tilde \nu_k = [\tilde \nu_{1,k}, \tilde \nu_{2,k},\nu_{3,k},\nu_{4,k}]^\top, $ as an additive disturbance satisfying $\|\tilde \nu_k\|_2\leq \epsilon_\nu$ with $\epsilon_\nu = 0.001$ to the discretized system.}

\subsection{Controller Design}

\rev{The robust MPC uses the constraint sets $\mathcal{X}:=\{x\mid \SI{0.05}{\mole\per\liter}\leq c_\gamma\leq \SI{0.45}{\mole\per\liter},\ \gamma\in\{A,B,C,D\}\}$ and
$\mathcal{U}:=\{u\mid \SI{1.00}{\mole\per\liter}\leq u\leq \SI{1.50}{\mole\per\liter}\}$,
together with the quadratic stage and terminal costs $l(x,u)=(x-x_\mathrm{ref})^\top Q(x-x_\mathrm{ref})+(u-u_\mathrm{ref})^\top R(u-u_\mathrm{ref})$
and $V_T(x)=(x-x_\mathrm{ref})^\top P(x-x_\mathrm{ref})$,
respectively, where $Q=\mathrm{diag}(2,1,1,1)$, $R=1$, and
\begin{align}
P=
\begin{bmatrix}
  8.77 & -5.29 &  3.27 & 1.06 \\
 -5.29 & 14.19 &  2.23 & 0.69 \\
  3.27 &  2.23 & 11.84 & 3.69 \\
  1.06 &  0.69 &  3.69 & 8.26
\end{bmatrix}.
\end{align}
The matrix $P$ is obtained from fixed-point iterations of the discrete-time algebraic Riccati equation for the model linearized at the reference. The terminal level is set to $\alpha_v=0.10$. We use the constant reference $x_\mathrm{ref} = [0.25,\,0.25,\,0.25,\,0.25]^\top \,\si{\mole\per\liter}$ and {$u_\mathrm{ref} = \SI{1.25}{\mole\per\liter}$}. We choose a prediction horizon of $N=10$ steps, corresponding to a prediction horizon of $\SI{5.00}{\minute}$. We conservatively set the designable approximation-error bound to $\epsilon=0.05$, corresponding to $\SI{10}{\percent}$ of the input range.}

\subsection{Neural Network Architecture}
\label{subsec:nnarc}

\rev{We approximate the underlying robust MPC law by a fully connected feed-forward neural network (cf. Section~\ref{subsec:ffnn}) with $3$ hidden layers, each containing $32$ neurons, and Gaussian error linear unit activation functions. The network takes the 4-dimensional concentration state as input and outputs the scalar feed concentration $c_{A,f}$.}
%A sigmoid output transformation enforces the input constraints $u\in[1,1.5]$.

\subsection{Data Generation and Training}
\label{subsec:data_training}

\rev{We seed $\num{194481}=21^4$ points equally spaced within the state constraint set. Using these points, we estimate a grid approximation $\hat{\mathcal{X}}_N\subseteq\mathcal{X}_N$ of the robustly forward-invariant set of the underlying MPC, which leaves $\num{60626}$ strictly feasible and converged grid points.}

\rev{Before dataset generation, we estimate the Lipschitz constant of the MPC policy at $\num{10000}$ points in $\hat{\mathcal{X}}_N$ using central finite differences with a step size of $\num{1e-4}$. This yields $L_{\mathrm{MPC}}=3.1531$. For the NN architecture described in
In Section~\ref{subsec:nnarc}, we prescribe the analytical design bound $L_{\mathrm{NN}}=3.5$, which is subsequently enforced via weight projection during training~\cite{gouk_regularisation_2021}. The projection approach is further explained in supplementary Section~\ref{supp:sec:weight-projection} and in \cite{gouk_regularisation_2021}. We choose the dataset error $\epsilon_{\mathcal D}=0.01$ and the overall approximation-error bound
$\epsilon=0.05$. Therefore, we can calculate $\delta = 0.0060122$, and we choose $\eta=\delta$. For the bounds and all subsequent analysis, the 2-norm was used.}

\rev{We construct a candidate grid over the bounding box of $\hat{\mathcal{X}}_N$ and retain points that pass a membership check with the invariant set. The final $\eta$-cover contains $\num{6417883}$ centers. The MPC policy is evaluated at every center to obtain the final training dataset $\mathcal D$. We then train a newly initialized NN on $\mathcal D$ for $50$ epochs
using a learning rate of $\num{5e-3}$ and a batch size of $128$. The
resulting NN achieves an $R^2$-score of $0.9999$ on a validation set containing $\num{10000}$ points that are distinct from all training points. The maximum approximation error over the complete training dataset is $\hat \epsilon_{\mathcal D} = 1.3759 \times 10^{-3} < 0.01$.
For the trained NNs Lipschitz bound obtained from the layer operator norms and activation Lipschitz constants is $\hat L_{\mathrm{NN}} = 3.4999$. Thus, both quantities remain below the used design values.}

\rev{All training, sampling, and simulations were performed in Python on a Lenovo ThinkPad T14 Gen 2a running Windows 11, equipped with an AMD Ryzen 7 PRO 5850U CPU and \(32\,\mathrm{GB}\) RAM. MPC was implemented using IPOPT via CasADi\cite{Casadi}, and the neural network using PyTorch\cite{pytorch1,pytorch2}.}

\rev{An additional 2-dimensional result is shown in supplementary Section~\ref{supp:sec:2dex}}

\subsection{Results}

\rev{We evaluate the trained NN on the simulated isothermal reactor in closed-loop using the robust MPC used for data generation as a comparison. The closed-loops for both controllers were simulated starting from $100$ independently selected initial states for $160$ steps with a sampling time of $\SI{0.5}{\minute}$. To assess the controller near the boundary of the estimated invariant set, points near the eight coordinate-wise extrema of the retained grid are included explicitly. The remaining initial states are sampled uniformly without replacement from $\hat{\mathcal{X}}_N$. None of the selected initial states were previously used for training or validation.  For each initial state, the robust MPC and NN controller experience the same additional disturbance realization.}

\rev{Fig.~\ref{fig:admissible-trajectories} shows the resulting trajectories in the $(c_A,c_B)$- and $(c_C,c_D)$-projections of $\hat{\mathcal{X}}_N$. The enlarged regions show trajectories initialized close to different boundaries of the estimated admissible set. The trajectories generated by the robust MPC and its NN approximation remain close over the complete simulation horizon and approach the reference.}
\rev{No state-constraint or input-constraint violations are observed for either controller. In particular, all \(\num{16000}\) evaluated states remain inside $\hat{\mathcal{X}}_N$. Set membership is evaluated in the full 4-dimensional state space and not solely from the 2-dimensional projections shown in Fig.~\ref{fig:admissible-trajectories}.}

\rev{Table~\ref{tab:isothermal-comptime} reports the computation times obtained from $\num{16000}$ sequential single-state controller evaluations. Both controllers are evaluated on the same hardware using a single CPU thread, and no batching is employed for the NN controller. The mean computation time decreases from $20.73\,\mathrm{ms}$ for the robust MPC to $0.59\,\mathrm{ms}$ for the NN controller, corresponding to a mean speed-up of approximately $35$ times.}

\rev{The pointwise approximation errors along the same independently generated trajectories are summarized in Table~\ref{tab:isothermal-errors}. The maximum observed error is $8.62\times 10^{-4}$,
which is below the prescribed approximation bound. No violation of the approximation bound occurs among the $16000$ evaluations. The minimum error is $5.73\times10^{-9}$ and is therefore close to the numerical precision relevant to this evaluation.}

\rev{These independent closed-loop results are consistent with the theoretical guarantees established through the underlying robust MPC, the coverage, and the certified approximation bound. In particular, the absence of approximation error bound, state, and input constraint violations supports the conditions of Theorem~\ref{thm:deltaeta} along all evaluated trajectories and demonstrates the applicability of Algorithm~\ref{alg:pipeline}.}

\begin{figure*}[!htb]
  \centering
  \begingroup
  \def\AdmissibleTikzDataDir{figures/tikz/tikz_data}
  \input{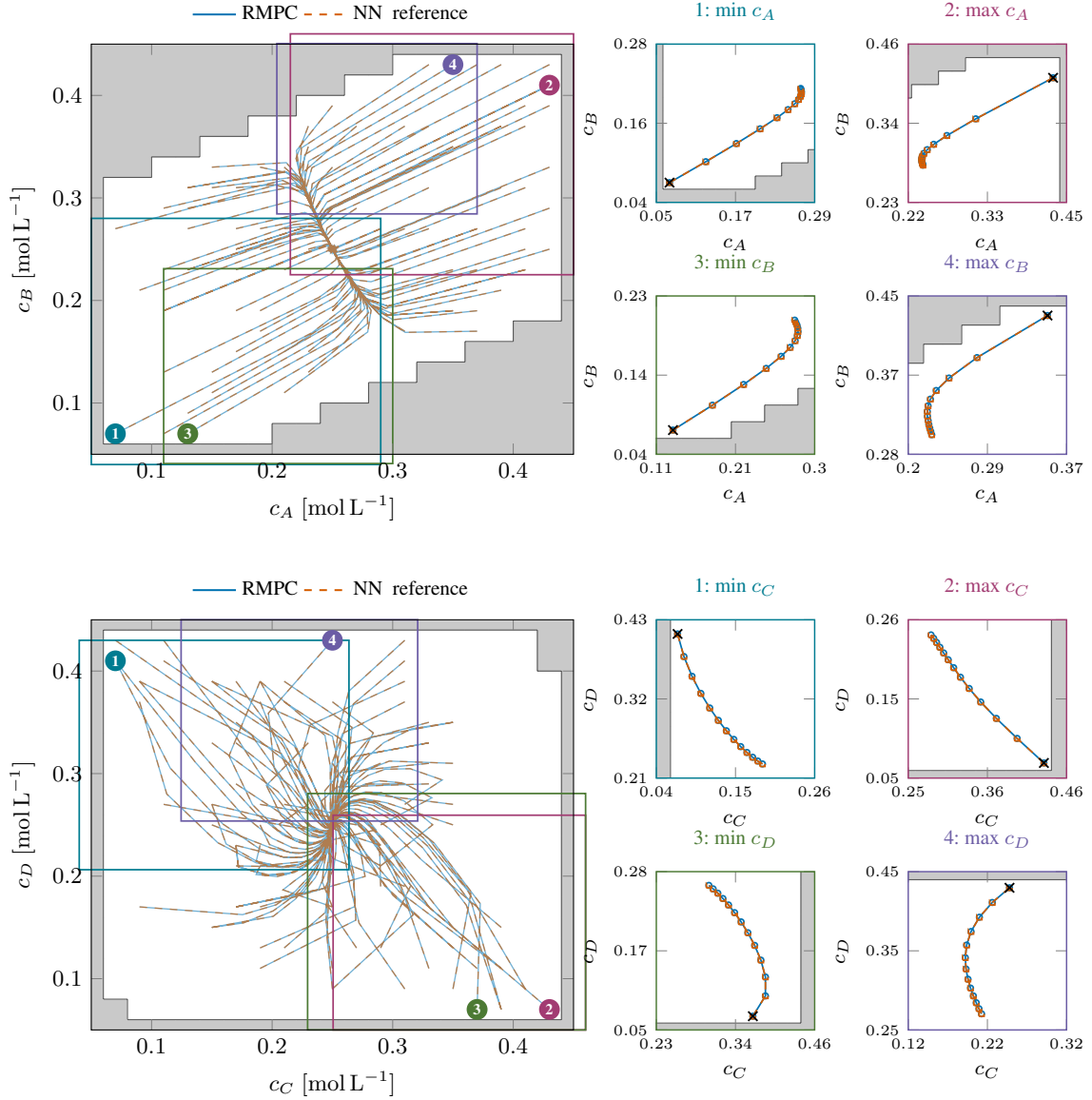}
  \endgroup
  \caption{Projections of the estimated invariant set $\hat{\mathcal{X}}_N$ and 100 closed-loop trajectories under the robust MPC and its NN approximation. Top: $(c_A, c_B)$-plane; bottom: $(c_C, c_D)$-plane. The insets show trajectories initialized near selected boundaries.}
  \label{fig:admissible-trajectories}
\end{figure*}

\begin{table}[bt]
    \centering
    \caption{\rev{Controller computation times over 16,000 evaluations.}}
    \small
    \setlength{\tabcolsep}{4pt}
    \begin{tabular}{|l|ccc|} \hline
         & \rev{Mean (ms)} & \rev{Max. (ms)} & \rev{Min. (ms)} \\ \hline
         \rev{Robust MPC}    & \rev{20.73} & \rev{85.79} & \rev{16.40} \\ \hline
         \rev{NN controller} & \rev{0.59}   & \rev{1.85}   & \rev{0.48}   \\ \hline
    \end{tabular}
    \label{tab:isothermal-comptime}
\end{table}

\begin{table}[bt]
    \centering
    \caption{\rev{Pointwise NN policy error over 16,000 evaluations.}}
    \small
    \setlength{\tabcolsep}{3pt}
    \begin{tabular}{|l|ccc|} \hline
         & \rev{Mean} & \rev{Max.} & \rev{Min.} \\ \hline
         \rev{Approximation error}
         & \rev{$4.76\!\times\!10^{-5}$}
         & \rev{$8.62\!\times\!10^{-4}$}
         & \rev{$5.73\!\times\!10^{-9}$} \\ \hline
    \end{tabular}
    \label{tab:isothermal-errors}
\end{table}

%%%%%%%%%%%%%%%%%%%%%%%%%%%%%%%%%%%%%%%%%%%%%%%%%%%%%%%%%%%%%%%%%%%%%%%%%%%%%%%%
\section{Conclusion}\label{sec:conc}

We developed a constructive framework for certified neural approximation of robust nonlinear model predictive controllers. By interpreting the approximation error as a bounded input disturbance and embedding Lipschitz-based error bounds into a robust MPC design, we established conditions under which recursive feasibility and input-to-state stability of the original controller are preserved under neural replacement. In contrast to post-verification approaches, the proposed method provides an offline design procedure that links dataset coverage and Lipschitz constants to explicit uniform approximation guarantees.

The framework combines results from parametric optimization, robust MPC theory, and Lipschitz continuity of neural networks to derive verifiable error bounds. The simulation study on a nonlinear isothermal CSRT model confirms that the neural controller reproduces the closed-loop behavior of the robust MPC within the certified invariant region while significantly reducing online computation time.

Future work \rev{will investigate} domain partitioning to reduce conservatism induced by global Lipschitz bounds and \rev{data requirements associated with uniform coverage, as well as adaptive or trajectory-based sampling and mixture-of-experts architectures to improve scalability to higher-dimensional systems.} 
\rev{We will also consider relaxing} the regularity assumptions on the MPC feedback law to broaden applicability \rev{of the proposed framework}.
%%%%%%%%%%%%%%%%%%%%%%%%%%%%%%%%%%%%%%%%%%%%%%%%%%%%%%%%%%%%%%%%%%%%%%%%%%%%%%%%
%\rev{\section*{Appendix}}

%%%%%%%%%%%%%%%%%%%%%%%%%%%%%%%%%%%%%%%%%%%%%%%%%%%%%%%%%%%%%%%%%%%%%%%%%%%%%%%%
%\section*{Acknowledgment}

%%%%%%%%%%%%%%%%%%%%%%%%%%%%%%%%%%%%%%%%%%%%%%%%%%%%%%%%%%%%%%%%%%%%%%%%%%%%%%%%

\section*{Declaration of Generative AI and AI-assisted technologies in the writing process}

During the preparation of this work, the authors used ChatGPT \rev{5.5} and \rev{5.6-Sol} from OpenAI in order to improve language flow and wording. After using this tool/service, the authors reviewed and edited the content as needed and take full responsibility for the content of the publication.

%%%%%%%%%%%%%%%%%%%%%%%%%%%%%%%%%%%%%%%%%%%%%%%%%%%%%%%%%%%%%%%%%%%%%%%%%%%%%%%%

%%%%%%%%%%%%%%%%%%%%%%%%%%%%%%%%%%%%%%%%%%%%%%%%%%%%%%%%%%%%%%%%%%%%%%%%%%%%%%%%
% SUPPLEMENTARY MATERIAL -- integrated into a single arXiv-compilable document.
% Labels are namespaced with supp: to avoid collisions with the article.
% Required: figures/coverage.pdf (export the original coverage.svg to PDF).
%%%%%%%%%%%%%%%%%%%%%%%%%%%%%%%%%%%%%%%%%%%%%%%%%%%%%%%%%%%%%%%%%%%%%%%%%%%%%%%%
\begingroup
\setcounter{section}{0}
\setcounter{subsection}{0}
\setcounter{figure}{0}
\setcounter{table}{0}
\setcounter{equation}{0}
\setcounter{thm}{0} % autart shares this counter across lemmas, theorems, etc.
\renewcommand{\thesection}{S\arabic{section}}
\renewcommand{\thesubsection}{\thesection.\arabic{subsection}}
\renewcommand{\thefigure}{S\arabic{figure}}
\renewcommand{\thetable}{S\arabic{table}}
\renewcommand{\theequation}{S\arabic{equation}}
\renewcommand{\thethm}{S\arabic{thm}}

\begin{center}
{\Large\bfseries Supplementary Material\par}
\vspace{1ex}
{\large Certified Neural Approximation of Robust Nonlinear Model Predictive Control\par}
\vspace{1ex}
Hendrik Alsmeier, Maik Pfefferkorn, Rolf Findeisen
\end{center}
\medskip
\noindent\textbf{Summary.}
This document supplements the article ``Certified Neural Approximation of Robust Nonlinear Model Predictive Control.'' Section S1 presents an additional two-dimensional numerical example. Section S2 describes the procedure used to enforce the prescribed neural-network Lipschitz bound during training. Section S3 provides the complete proof of Lemma 9 in the main article.

\section{2-Dimensional Example}\label{supp:sec:2dex} 

In addition to the continuous stirred tank reactor example presented in the main part of the paper, we further evaluate the approach on a nonlinear drone model describing vertical motion. The state $x = [x_1, x_2] = [z, v_z]$ consists of altitude and vertical velocity, and the dynamics are
given by
\begin{align}
    x_{1,k+1} &= x_{1,k} + T_s\, x_{2,k} + \tilde \nu_{1,k}, \\
    x_{2,k+1} &= x_{2,k} + T_s\left(-g - k_{z}\, x_{2,k}^3 + \frac{1}{m}u_k\right) + \tilde \nu_{2,k}.
\end{align}
Here, $u$ is the input thrust force, $k_z=\SI{0.0001}{\second\per\meter\squared}$ is the cubic drag coefficient in $z$-direction, $m=\SI{1}{\kg}$ is the mass of the drone, $g =\SI{9.81}{\meter\per\square\second}$ is the acceleration due to gravity, and $T_s = 0.05$ is the discretization time. Additionally, $\tilde \nu_k = [\tilde \nu_{1,k}, \tilde \nu_{2,k}]^\top$ denotes an additive disturbance satisfying $\|\tilde \nu_k\|_2\leq \epsilon_\nu$ with $\epsilon_\nu = 0.001$. The undisturbed system therefore has Lipschitz constants $L_{f,x} = 1.05$ and $L_{f,u} = 0.05$ w.r.t. the $2$-norm, which can be calculated directly due to its control-affine structure.

\subsection{Controller Design}

For the robust MPC, we define the state constraint set $\mathcal{X}$ by $\SI{-20}{\meter} \leq x_1 \leq \SI{40}{\meter}$ and $\SI{-25}{\meter\per\second}  \leq x_2 \leq \SI{25}{\meter\per\second}$, with positive velocity pointing upward, and the input constraint set $\mathcal{U}$ by $\SI{-30}{\N}  \leq u \leq \SI{30}{\N}$.

We use a quadratic stage cost $l(x,u) = (x-x_\mathrm{ref})^\top Q (x-x_\mathrm{ref}) + (u-u_\mathrm{ref})^\top R (u-u_\mathrm{ref})$, where $Q = \mathrm{diag}(5,1)$ and $R = 0.1$.
For the terminal cost, we choose $V_T(x) = (x-x_\mathrm{ref})^\top P (x-x_\mathrm{ref})$ with
\begin{align}
P=\begin{bmatrix} 74.6404 & 15.9899 \\ 15.9899 & 11.1355 \end{bmatrix}
\end{align}
obtained via fixed-point iterations of the discrete-time algebraic Riccati equation. %, and we fix $\alpha_v = 100$.
The terminal set is defined as the $\alpha_v$-level set of $V_T$ with $\alpha_v = 100$.
We set a constant reference of $x_\mathrm{ref} = [\SI{25.0}{m},\SI{0}{\meter\per\second}]^\top$ and $u_\mathrm{ref} = g \cdot m = \SI{9.81}{\N}$.

% We conservatively set the designable error bound $\epsilon = 3$, which is justifiable compared to the possible range of $u$.
The designable approximation error bound is set to $\epsilon = 3$, which is small relative to the admissible range of $u$.

\subsection{Horizon length and conservatism} For the robust MPC formulation in \eqref{eq:ocp}, the disturbance sets $W_n$ may grow with the prediction step. Their radius $r_n$ is
\begin{equation}
r_n =
\frac{L_{f,x}^{n}-1}{L_{f,x}-1}L_{f,u}(\epsilon+\epsilon_\nu).
\end{equation}
Hence, $W_n$ remains bounded for $L_{f,x}<1$, grows linearly for $L_{f,x}=1$, and may grow geometrically for $L_{f,x}>1$. \\

For the drone model, $L_{f,x}=1.05>1$, and the disturbance sets therefore grow along the prediction horizon. With \(L_{f,u}=0.05\), \(\epsilon=3\), and a chosen $N=30$, the radius of $W_{N-1}=W_{29}$ is approximately
\begin{equation}
r_{29}
=
\frac{1.05^{29}-1}{1.05-1}(0.05 \cdot 3+0.001)
\approx 9.41.
\end{equation}
For comparison, the state domain has widths $60$ and $50$ in the altitude and velocity directions, respectively. Under the 2-norm used in the example, the tightening by $W_{29}$ reduces these widths by \(2r_{29}\approx18.7\), corresponding to reductions of approximately $\SI{31}{\percent}$ and $\SI{37}{\percent}$, respectively. This illustrates that the horizon-dependent tightening is already substantial for $N=30$. We therefore selected this horizon as a practical trade-off between improved look-ahead and increased conservatism. \\

A tube-based MPC formulation with an ancillary feedback controller could modify this propagation by stabilizing the error dynamics.
Under suitable conditions ensuring contraction of the resulting closed-loop error dynamics, the corresponding uncertainty sets could remain bounded or grow substantially more slowly.
However, the use of an ancillary feedback controller alone does not guarantee such behavior; establishing boundedness or contraction would require additional assumptions and analysis beyond the scope of the present framework.

\subsection{Neural Network Architecture}
\label{supp:subsec:nnarc}

We approximate the underlying robust MPC law by a fully connected feed-forward neural network (cf. Section~\ref{subsec:ffnn}) with $3$ hidden layers, each containing $16$ neurons and using a hyperbolic tangent activation function. The network takes the state $x$ as input and outputs the control input $u$. This architecture is used for the initial and final training phases described in Section~\ref{supp:subsec:data_training}.

\subsection{Data Generation and Training}
\label{supp:subsec:data_training}

We seed $\num{20814}$ initial points equally spaced within the state bounds. Using these points, we estimate a robustly forward-invariant set $\hat{\mathcal{X}}_N \subseteq \mathcal{X}_N$ of the underlying MPC using a backwards-pruning approach.

With the same seed points, we perform an initial training of the NN architecture (cf. Section~\ref{supp:subsec:nnarc}) for $200$ epochs with a learning rate of $\num{5e-3}$, which yields an approximation error of $\epsilon_\mathcal{D}=1.9876$, rounded to $\epsilon_\mathcal{D}=2$ in the $2$-norm.

Before dataset generation begins, we estimate the Lipschitz constants of both the MPC and the neural network on $\hat{\mathcal{X}}_N$ using a power-iteration method based on repeated Jacobian–vector and vector–Jacobian products. For the MPC, we obtain $L_\mathrm{MPC} = 19.887$, which we round to $L_\mathrm{MPC} = 20$, and for the initially trained neural
network, we obtain $21.436$, from which we chose $L_\mathrm{NN} = 25$ for caution. This results in $\delta = 0.0222$, based on which we set $\eta = \delta$. We then compute an $\eta$-coverage of $\mathcal{X}_N$ by calculating a bounding box of $\hat{\mathcal{X}}_N$ on which a space-filling sampling algorithm is employed.

This results in $1{\,}231{\,}052$ coverage points $x_i$, which are
used to construct the final training dataset $\mathcal D$. These points provide an $\eta$-coverage of $\hat{\mathcal X}_N$, i.e.,
\begin{equation}
    \sup_{x\in\hat{\mathcal X}_N}
    \min_{x_i\in\mathcal D}
    \lVert x-x_i\rVert_2
    \leq \eta.
\end{equation}
Equivalently, the balls of radius $\eta$ centered at the coverage
points cover $\hat{\mathcal X}_N$. Consequently, two neighboring
coverage points whose corresponding balls overlap are separated by
at most \(2\eta\).
This bound applies to neighboring points; it does not impose a distance bound on arbitrary pairs of coverage points. A cutout of the cover is shown in Fig.~\ref{supp:fig:cover}

\begin{figure}[!hb]
  \centering
  \includegraphics[width=\linewidth]{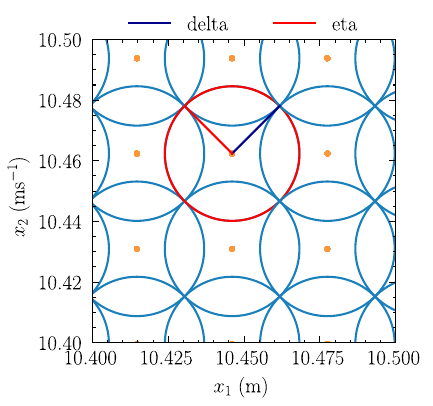}
  \caption{Small cutout of $\hat{\mathcal{X}}_N$ in state space with coverage. Original coverage radius $\delta$ as a black line and red circle. Over-coverage radius $\eta$ as a red line and blue circles. Data points $x_i$ are represented as orange points.}
  \label{supp:fig:cover}
\end{figure}

We then train a newly initialized NN on the dataset $\mathcal{D}$. The training concludes after $100$ epochs with a learning rate of $\num{5e-3}$ and a batch size of $\num{1080}$, achieving an $R^2$-score of $0.9999$ calculated on a validation set of $\num{10000}$ data points, which are distinct from any training data points. Moreover, the trained NN achieves $\hat \epsilon_\mathcal{D} = 1.9222$ and a Lipschitz estimate $L_\mathrm{NN}^{\mathrm{an,final}} = 23.22$, both of which satisfy the initial bounds.

All training, sampling, and simulations were performed in Python on a Lenovo ThinkPad T14 Gen 2a running Windows 11, equipped with an AMD Ryzen 7 PRO 5850U CPU and \(32\,\mathrm{GB}\) RAM. MPC was implemented using IPOPT via CasADi, and the neural network using PyTorch.

\subsection{Results}

We simulate the system in closed-loop, controlled by the robust MPC and its NN approximation for 100 random initial states. 
Specifically, the initial states are drawn from a uniform distribution over $\hat{\mathcal{X}}_N$.
The closed-loop system is simulated for $15\,\mathrm{s}$.
%The resulting trajectories are shown in Fig.~\ref{supp:fig:cl_t} and \ref{supp:fig:cl_stsp}. 
We observe that the state trajectories under the robust MPC and the NN controller remain close over the whole simulation time and all converge to a neighborhood of the reference point (Figure \ref{supp:fig:admissible-trajectories}). Furthermore, we can observe that no constraint violations occur, and that all trajectories stay inside $\hat{\mathcal{X}}_N$.
Table~\ref{supp:tab:comptime} shows the mean, maximum, and minimum computation times for both the MPC and NN controller along the 100 trajectories, indicating a significant speed-up using the NN controller. For the same trajectories, Table~\ref{supp:tab:errors} presents the occurring approximation error in $l_2$-norm, all of which are below our allowed bound of $\epsilon = 3$. The minimal value is given as $0$ since the first observed digit is at the seventh decimal point and therefore not clearly distinguishable from numerical noise.

The results show how safety and stability guarantees established via the underlying MPC transfer to the neural network controller.
Furthermore, the results suggest that Theorem \ref{thm:closed_loop_guarantees} is satisfied along all 100 closed-loop trajectories, indicating that the feedback law of the designed MPC is well-behaved and that the proposed Algorithm \ref{alg:pipeline} is suited to compute the NN approximator.

\begin{table}[!bt]
    \centering
    \caption{Computation time for the robust MPC and NN controller.}
    \begin{tabular}{|l|ccc|} \hline 
         & Mean (ms) & Max (ms) & Min (ms) \\ \hline
         MPC        & 41.79  & 213.98  & 24.65 \\ \hline
         NN approx. & 0.52   & 1.98   & 0.44 \\ \hline
    \end{tabular}
    \label{supp:tab:comptime}
\end{table}

\begin{table}[!hbt]
    \centering
    \caption{Approximation error over 100 trajectories for the robust MPC and NN controller.}
    \begin{tabular}{|l|ccc|} \hline 
         & Mean & Max & Min \\ \hline
         Approx. errors & $0.1943$ & $1.3497$  & $7.5449\!\times\!10^{-6}$ \\ \hline
    \end{tabular}
    \label{supp:tab:errors}
\end{table}

\begin{figure*}[!ht]
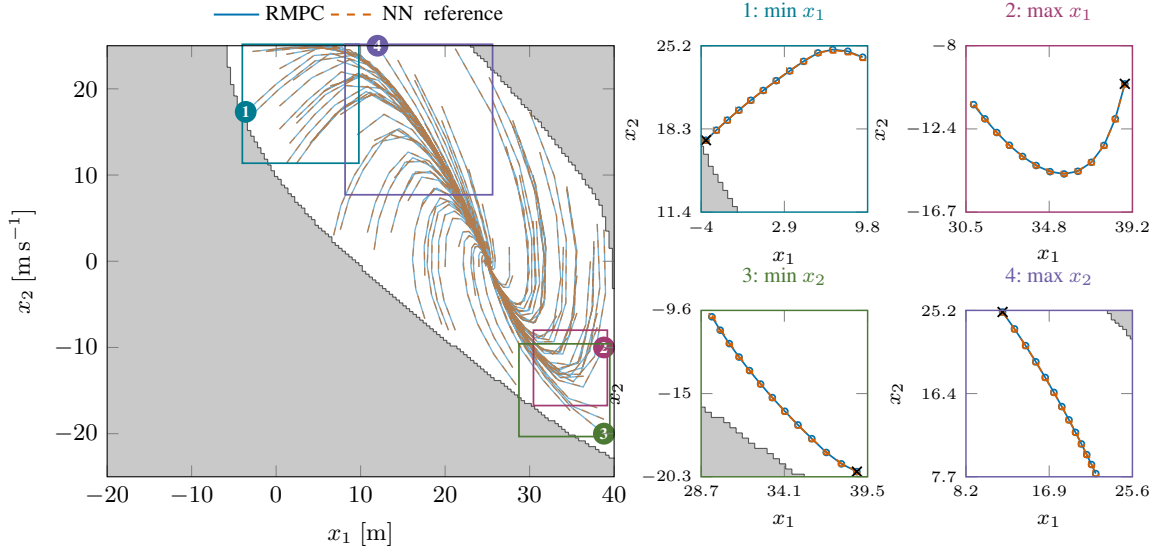

  \centering
  \begingroup
  \def\AdmissibleTikzDataDir{figures/tikz/tikz_data}
  % Generated by export_admissible_tikz_drone.py.
% Define \AdmissibleTikzDataDir before \input when needed.
\providecommand{\AdmissibleTikzDataDir}{tikz_data}
\definecolor{rmpcblue}{HTML}{0072B2}
\definecolor{nnorange}{HTML}{D55E00}
\definecolor{outsidegray}{HTML}{C9C9C9}
\definecolor{violationred}{HTML}{C62828}
\definecolor{zoomone}{HTML}{007C91}
\definecolor{zoomtwo}{HTML}{A23B72}
\definecolor{zoomthree}{HTML}{4C7A34}
\definecolor{zoomfour}{HTML}{6B5CA5}
\pgfplotsset{
  colormap={admissiblemap}{
    color(0cm)=(outsidegray);
    color(1cm)=(white)
  },
  admissible main axis/.style={
    width=6.7cm,
    height=5.7cm,
    scale only axis,
    axis on top,
    clip mode=individual,
    tick label style={font=\footnotesize},
    label style={font=\small},
    legend style={
      at={(0.5,1.025)},
      anchor=south,
      draw=none,
      fill=none,
      font=\scriptsize,
      legend columns=3
    },
    legend image post style={opacity=1,line width=0.75pt},
    scaled ticks=false
  },
  admissible zoom axis/.style={
    width=2.2cm,
    height=2.2cm,
    scale only axis,
    axis equal image,
    axis on top,
    clip mode=individual,
    tick label style={font=\tiny},
    label style={font=\scriptsize},
    title style={font=\scriptsize,align=center},
    scaled ticks=false,
    xticklabel style={/pgf/number format/fixed,/pgf/number format/precision=1},
    yticklabel style={/pgf/number format/fixed,/pgf/number format/precision=1},
    max space between ticks=28
  },
  admissible boundary/.style={
    color=black!65,
    line width=0.42pt,
    unbounded coords=jump
  },
  rmpc trajectories/.style={
    color=rmpcblue,
    opacity=0.55,
    line width=0.50pt,
    unbounded coords=jump
  },
  nn trajectories/.style={
    color=nnorange,
    opacity=0.62,
    dashed,
    line width=0.50pt,
    unbounded coords=jump
  },
  rmpc zoom/.style={
    color=rmpcblue,
    mark=o,
    mark size=1.1pt,
    line width=0.65pt
  },
  nn zoom/.style={
    color=nnorange,
    dashed,
    mark=square,
    mark options={fill=white},
    mark size=1.0pt,
    line width=0.65pt
  },
  violation marks/.style={
    only marks,
    color=violationred,
    mark=x,
    mark size=2.3pt
  },
  violation rings/.style={
    only marks,
    color=violationred,
    mark=o,
    mark options={fill=none},
    mark size=2.3pt
  }
}
\tikzset{
  zoom number/.style={
    circle,
    fill=#1,
    text=white,
    font=\bfseries\tiny,
    inner sep=1.1pt
  }
}
\begin{tikzpicture}
\begin{axis}[
  admissible main axis,
  name=mainDrone,
  at={(0cm,0cm)},
  anchor=south west,
  xmin=-20,
  xmax=40,
  ymin=-25,
  ymax=25,
  xlabel={$x_1\;[\mathrm{m}]$},
  ylabel={$x_2\;[\mathrm{m\,s^{-1}}]$}
]
  \begin{scope}
  \clip (axis cs:-20,-25) rectangle (axis cs:40,25);
  \input{\AdmissibleTikzDataDir/drone_background.tikz}
  \end{scope}
  \addplot[admissible boundary,forget plot] table[x=x,y=y,col sep=comma] {\AdmissibleTikzDataDir/drone_boundary.csv};
  \addplot[rmpc trajectories] table[x=x,y=y,col sep=comma] {\AdmissibleTikzDataDir/drone_rmpc.csv};
  \addlegendentry{RMPC}
  \addplot[nn trajectories] table[x=x,y=y,col sep=comma] {\AdmissibleTikzDataDir/drone_nn.csv};
  \addlegendentry{NN}
  \addplot[only marks,mark=star*,mark size=2.5pt,black] coordinates {(25,0)};
  \addlegendentry{reference}
  \draw[zoomone,line width=0.65pt] (axis cs:-4,11.3724193585) rectangle (axis cs:9.79424730819,25.1666666667);
  \node[zoom number=zoomone] at (axis cs:-3.6,17.3333333333) {1};
  \draw[zoomtwo,line width=0.65pt] (axis cs:30.4566771726,-16.7385711682) rectangle (axis cs:39.2,-7.99524834074);
  \node[zoom number=zoomtwo] at (axis cs:38.8,-10) {2};
  \draw[zoomthree,line width=0.65pt] (axis cs:28.7387778565,-20.3333333333) rectangle (axis cs:39.4915984889,-9.58051270094);
  \node[zoom number=zoomthree] at (axis cs:38.8,-20) {3};
  \draw[zoomfour,line width=0.65pt] (axis cs:8.16172535361,7.71154353833) rectangle (axis cs:25.616848482,25.1666666667);
  \node[zoom number=zoomfour] at (axis cs:12,25) {4};
\end{axis}
\begin{axis}[
  admissible zoom axis,
  name=zoomDrone1,
  at={(7.85cm,3.50cm)},
  anchor=south west,
  xmin=-4,
  xmax=9.79424730819,
  ymin=11.3724193585,
  ymax=25.1666666667,
  xtick={-4,2.8971236541,9.79424730819},
  ytick={11.3724193585,18.2695430126,25.1666666667},
  xlabel={$x_1$},
  ylabel={$x_2$},
  axis line style={color=zoomone,line width=0.55pt},
  title style={text=zoomone},
  title={1: min $x_1$}
]
  \begin{scope}
  \clip (axis cs:-4,11.3724193585) rectangle (axis cs:9.79424730819,25.1666666667);
  \input{\AdmissibleTikzDataDir/drone_background.tikz}
  \end{scope}
  \addplot[admissible boundary,forget plot] table[x=x,y=y,col sep=comma] {\AdmissibleTikzDataDir/drone_boundary.csv};
  \addplot[rmpc zoom] table[x=mpc_x,y=mpc_y,col sep=comma] {\AdmissibleTikzDataDir/drone_zoom_1.csv};
  \addplot[nn zoom] table[x=nn_x,y=nn_y,col sep=comma] {\AdmissibleTikzDataDir/drone_zoom_1.csv};
  \addplot[only marks,mark=x,mark size=2.6pt,black,line width=0.7pt] coordinates {(-3.6,17.3333333333)};
\end{axis}
\begin{axis}[
  admissible zoom axis,
  name=zoomDrone2,
  at={(11.35cm,3.50cm)},
  anchor=south west,
  xmin=30.4566771726,
  xmax=39.2,
  ymin=-16.7385711682,
  ymax=-7.99524834074,
  xtick={30.4566771726,34.8283385863,39.2},
  ytick={-16.7385711682,-12.3669097545,-7.99524834074},
  xlabel={$x_1$},
  ylabel={$x_2$},
  axis line style={color=zoomtwo,line width=0.55pt},
  title style={text=zoomtwo},
  title={2: max $x_1$}
]
  \begin{scope}
  \clip (axis cs:30.4566771726,-16.7385711682) rectangle (axis cs:39.2,-7.99524834074);
  \input{\AdmissibleTikzDataDir/drone_background.tikz}
  \end{scope}
  \addplot[admissible boundary,forget plot] table[x=x,y=y,col sep=comma] {\AdmissibleTikzDataDir/drone_boundary.csv};
  \addplot[rmpc zoom] table[x=mpc_x,y=mpc_y,col sep=comma] {\AdmissibleTikzDataDir/drone_zoom_2.csv};
  \addplot[nn zoom] table[x=nn_x,y=nn_y,col sep=comma] {\AdmissibleTikzDataDir/drone_zoom_2.csv};
  \addplot[only marks,mark=x,mark size=2.6pt,black,line width=0.7pt] coordinates {(38.8,-10)};
\end{axis}
\begin{axis}[
  admissible zoom axis,
  name=zoomDrone3,
  at={(7.85cm,0.00cm)},
  anchor=south west,
  xmin=28.7387778565,
  xmax=39.4915984889,
  ymin=-20.3333333333,
  ymax=-9.58051270094,
  xtick={28.7387778565,34.1151881727,39.4915984889},
  ytick={-20.3333333333,-14.9569230171,-9.58051270094},
  xlabel={$x_1$},
  ylabel={$x_2$},
  axis line style={color=zoomthree,line width=0.55pt},
  title style={text=zoomthree},
  title={3: min $x_2$}
]
  \begin{scope}
  \clip (axis cs:28.7387778565,-20.3333333333) rectangle (axis cs:39.4915984889,-9.58051270094);
  \input{\AdmissibleTikzDataDir/drone_background.tikz}
  \end{scope}
  \addplot[admissible boundary,forget plot] table[x=x,y=y,col sep=comma] {\AdmissibleTikzDataDir/drone_boundary.csv};
  \addplot[rmpc zoom] table[x=mpc_x,y=mpc_y,col sep=comma] {\AdmissibleTikzDataDir/drone_zoom_3.csv};
  \addplot[nn zoom] table[x=nn_x,y=nn_y,col sep=comma] {\AdmissibleTikzDataDir/drone_zoom_3.csv};
  \addplot[only marks,mark=x,mark size=2.6pt,black,line width=0.7pt] coordinates {(38.8,-20)};
\end{axis}
\begin{axis}[
  admissible zoom axis,
  name=zoomDrone4,
  at={(11.35cm,0.00cm)},
  anchor=south west,
  xmin=8.16172535361,
  xmax=25.616848482,
  ymin=7.71154353833,
  ymax=25.1666666667,
  xtick={8.16172535361,16.8892869178,25.616848482},
  ytick={7.71154353833,16.4391051025,25.1666666667},
  xlabel={$x_1$},
  ylabel={$x_2$},
  axis line style={color=zoomfour,line width=0.55pt},
  title style={text=zoomfour},
  title={4: max $x_2$}
]
  \begin{scope}
  \clip (axis cs:8.16172535361,7.71154353833) rectangle (axis cs:25.616848482,25.1666666667);
  \input{\AdmissibleTikzDataDir/drone_background.tikz}
  \end{scope}
  \addplot[admissible boundary,forget plot] table[x=x,y=y,col sep=comma] {\AdmissibleTikzDataDir/drone_boundary.csv};
  \addplot[rmpc zoom] table[x=mpc_x,y=mpc_y,col sep=comma] {\AdmissibleTikzDataDir/drone_zoom_4.csv};
  \addplot[nn zoom] table[x=nn_x,y=nn_y,col sep=comma] {\AdmissibleTikzDataDir/drone_zoom_4.csv};
  \addplot[only marks,mark=x,mark size=2.6pt,black,line width=0.7pt] coordinates {(12,25)};
\end{axis}
\end{tikzpicture}
  \endgroup
   \caption{100 closed-loop state trajectories of the robust MPC and the NN controllers in state space with initial states in $\hat{\mathcal{X}}_N$,  shown in white. The four insets provide enlarged views of selected trajectories for clarity.}
  \label{supp:fig:admissible-trajectories}
\end{figure*}

\section{Weight-projected neural network training}
\label{supp:sec:weight-projection}

The Lipschitz constant of the feed-forward neural network introduced
in Section~\ref{subsec:ffnn} of the main article can be bounded using the Lipschitz
constants of its individual layers. Let \(L_{\sigma_j}\) denote a
Lipschitz constant of the activation function \(\sigma_j\), and let
\(A_j\) denote the weight matrix of the \(j\)-th affine layer. Using
the submultiplicativity of induced matrix norms and the composition of
the network layers, we obtain the analytical upper bound
\begin{equation}
    \hat L_{\mathrm{NN}}(\theta)
    =
    \lVert A_m\rVert_2
    \prod_{j=1}^{m-1}
    \left(
        L_{\sigma_j}\lVert A_j\rVert_2
    \right),
    \label{supp:eq:nn-lipschitz-bound}
\end{equation}
where \(\theta\) denotes the collection of all network parameters and
\(m\) is the number of weight matrices. Fixed linear input or output
transformations contribute additional induced-norm factors to
\eqref{supp:eq:nn-lipschitz-bound}. Equivalently, they can be incorporated
into the first or last affine map when evaluating the bound. Since the
bias vectors \(b_j\) are constant additional factors and therefore do not affect differences between network outputs, they do not contribute to
\(\hat L_{\mathrm{NN}}\)
\cite{gouk_regularisation_2021}.

To enforce a prescribed bound \(L_{\mathrm{NN}}>0\),
the network is trained using projected gradient-based optimization.
Starting from the parameters \(\theta^{(k)}\), an unconstrained
optimizer step with respect to the supervised loss
\(\mathcal L\) first produces tentative parameters
\(\tilde\theta^{(k+1)}\). These parameters are subsequently mapped
back to the feasible parameter set
\begin{equation}
    \Theta_L
    :=
    \left\{
        \theta
        \,\middle|\,
        \widehat L_{\mathrm{NN}}(\theta)
        \leq L_{\mathrm{NN}}
    \right\}.
\end{equation}

We employ a uniform radial rescaling of all \(m\) weight matrices. The
scaling factor is defined as
\begin{equation}
    \gamma^{(k+1)}
    =
    \begin{cases}
        \displaystyle
        \left(
            \frac{L_{\mathrm{NN}}}
            {\hat L_{\mathrm{NN}}
                \bigl(\tilde\theta^{(k+1)}\bigr)}
        \right)^{1/m},
        &
        \hat L_{\mathrm{NN}}
            \bigl(\widetilde\theta^{(k+1)}\bigr)
        >
        L_{\mathrm{NN}},
        \\[4mm]
        1,
        &
        \text{otherwise}.
    \end{cases}
    \label{supp:eq:nn-projection-factor}
\end{equation}
The updated parameters are then given by
\begin{align}
    A_j^{(k+1)}
    &=
    \gamma^{(k+1)}
    \tilde A_j^{(k+1)},
    && j=1,\ldots,m,
    \label{supp:eq:nn-radial-projection}\\
    b_j^{(k+1)}
    &=
    \tilde b_j^{(k+1)},
    && j=1,\ldots,m.
\end{align}
Because induced matrix norms are positively homogeneous, the resulting
bound satisfies
\begin{equation}
    \hat L_{\mathrm{NN}}\bigl(\theta^{(k+1)}\bigr)
    =
    \bigl(\gamma^{(k+1)}\bigr)^m
    \hat L_{\mathrm{NN}}
        \bigl(\tilde\theta^{(k+1)}\bigr)
    \leq
    L_{\mathrm{NN}}.
\end{equation}
Thus, every completed optimization step yields parameters that satisfy
the prescribed analytical Lipschitz bound~\cite{gouk_regularisation_2021}.

A soft penalty on violations of
\(L_{\mathrm{NN}}\) may additionally be included in the
training objective to reduce the magnitude and frequency of the
rescaling steps. Related projection-based approaches are discussed in
greater detail in \cite{gouk_regularisation_2021}.

\section{Existence of a neural network with prescribed finite-dataset accuracy}
\label{supp:sec:finite-dataset-approximation}

This section provides the complete proof of Lemma~\ref{lem:nn_interpolates_finite} in the main article.
Specifically, we show that, for any finite dataset with pairwise distinct
inputs and any prescribed tolerance \(\epsilon_{\mathcal D}>0\), there
exists a feed-forward neural network whose approximation error at every
data point is at most \(\epsilon_{\mathcal D}\). The result concerns
existence and finite-dataset accuracy only. This statement does not imply that a
particular training algorithm finds such a network or that the resulting
network satisfies an approximation bound outside the training dataset.

\begin{lem}(Existence of an $\epsilon_{\mathcal D}$-interpolating NN.)\label{supp:lem:nn_interpolates_finite}
Let $\mathcal{D} = \{ (x_i, y_i) \in \mathbb{R}^d \times \mathbb{R}^p \mid i \in \{1, \ldots, n_{\mathcal{D}} \}\}$ be a dataset with pairwise distinct inputs \(x_i\) and $\epsilon_{\mathcal D}>0$ a given constant. Then there exists a feedforward neural network 
$\tilde f:\mathbb{R}^d\to\mathbb{R}^p$ with any continuous non-polynomial activation function such that $\|\tilde f(x_i)-y_i\|\le \epsilon_{\mathcal D}$ for all $i=1,\dots,n_{\mathcal{D}}$.
\end{lem}

\begin{pf}
The proof consists of two steps.
First, we show that there exists a continuous function that interpolates the training data.
Second, we apply the universal approximation theorem to conclude the existence of a neural network that approximates this interpolant with $\epsilon_\mathcal{D}$ accuracy.
As a result, this neural network satisfies $\|\tilde f(x_i)-y_i\|\le \epsilon_{\mathcal D}$ point-wise on the training data, which proves the claim.

\textbf{Step 1:} 
Consider the set $\mathcal{S}=\{x_i \mid i \in \{1, \ldots, n_{\mathcal{D}} \} \}  \subset\mathbb{R}^d$ consisting of the training inputs. 
Since $\mathcal{S}$ is finite, it is closed.
Now, define the function $g:S\to\mathbb{R}^p$ by $g(x_i)=y_i$ for all $i=1, \ldots, n_\mathcal{D}$. 
We now aim to construct a continuous extension of the function $g$ from the set $\mathcal{S}$ to the entire space $\mathbb{R}^d$, yielding the desired interpolant.
To this end, we consider the component functions $g_j:S\to\mathbb{R}$, such that $g(x) = (g_1(x), \ldots, g_p(x))$, and each component satisfies $g_j(x_i)=y_i^{(j)}$, where $y_i=(y_i^{(1)},\dots,y_i^{(p)})$.
As $\mathbb{R}^d$, equipped with any $p$-norm, is a normal topological space and $S$ is closed, the Tietze extension theorem \cite[Theorem~12.4]{jameson1974} guarantees that each real-valued component function $g_j$ can be continuously extended from $\mathcal{S}$ to $\mathbb{R}^d$.
We denote the continuous extension of $g_j$ by $q_j$, and define the continuous interpolant $q:\mathbb{R}^d\to\mathbb{R}^p$ of the training data by $q(x) = (q_1(x), \ldots, q_p(x))$.

\textbf{Step 2:}
We consider a compact set $\mathcal{K}\subset\mathbb{R}^d$ of interest such that $\mathcal{S}\subset \mathcal{K}$.
By the universal approximation theorem \cite{hornik1991}, for any $\delta'>0$, there exists a neural network $\tilde{f}:\mathbb{R}^d\to\mathbb{R}^p$ such that $\sup_{x \in \mathcal{K}} \| \tilde{f}(x)-q(x) \|_\infty < \delta'$. 
The result can be extended to arbitrary $p$-norms with $1\leq p < \infty$ by using norm equivalences in finite-dimensional real spaces.
Specifically, for any $1 \leq p < \infty$, there exists a constant $c_p \geq 1$ such that $\|v\|_p \le c_p\|v\|_\infty$ for all $v\in\mathbb{R}^p$, and trivially, $c_p=1$ if $p=\infty$.
% Let $\|\cdot\|$ be the norm used in the statement on $\mathbb{R}^p$.
% Since norms are equivalent in finite dimensions, there is $c_p\ge1$ with
% $\|v\|\le c_p\|v\|_\infty$ for all $v\in\mathbb{R}^p$.
Thus, and taking $\delta'=\epsilon_{\mathcal D}/c_p$, we obtain
$\sup_{x\ in K}\|\tilde f(x)-q(x)\|<\epsilon_{\mathcal D}$, and in particular,
$\|\tilde f(x_i)-y_i\|\le\epsilon_{\mathcal D}$ for all $i$, which concludes the proof.
\end{pf}

%%%%%%%%%%%%%%%%%%%%%%%%%%%%%%%%%%%%%%%%%%%%%%%%%%%%%%%%%%%%%%%%%%%%%%%%%%%%%%%%%%%%%%%%%%%%%%%%%%%

\endgroup

%%%%%%%%%%%%%%%%%%%%%%%%%%%%%%%%%%%%%%%%%%%%%%%%%%%%%%%%%%%%%%%%%%%%%%%%%%%%%%%%
% Single shared bibliography for article and supplement.
%%%%%%%%%%%%%%%%%%%%%%%%%%%%%%%%%%%%%%%%%%%%%%%%%%%%%%%%%%%%%%%%%%%%%%%%%%%%%%%%
\bibliographystyle{plain}
\bibliography{literature}

\end{document}